\documentclass[11pt]{article}

\def\showauthornotes{0}
\def\showtableofcontents{1}
\def\showkeys{0}
\def\showdraftbox{0}
\def\showcolorlinks{1}
\def\usemicrotype{0}
\def\showfixme{0}

\usepackage{etex}

\usepackage{xspace,enumerate}

\usepackage[dvipsnames]{xcolor}

\usepackage[T1]{fontenc}
\usepackage[full]{textcomp}

\usepackage[american]{babel}

\usepackage{mathtools}

\usepackage{bm}

\usepackage{amsthm}

\newtheorem{theorem}{Theorem}[section]
\newtheorem*{theorem*}{Theorem}
\newtheorem{itheorem}{Theorem}

\newtheorem*{claim*}{Claim}

\newtheorem*{proposition*}{Proposition}
\newtheorem{lemma}[theorem]{Lemma}
\newtheorem*{lemma*}{Lemma}
\newtheorem{corollary}[theorem]{Corollary}
\newtheorem{conjecture}[theorem]{Conjecture}
\newtheorem*{conjecture*}{Conjecture}
\newtheorem{observation}[theorem]{Observation}
\newtheorem{fact}[theorem]{Fact}
\newtheorem*{fact*}{Fact}

\newtheorem*{hypothesis*}{Hypothesis}

\theoremstyle{definition}
\newtheorem{definition}[theorem]{Definition}

\newtheorem{remark}[theorem]{Remark}

\usepackage[letterpaper,
top=1.3in,
bottom=1.3in,
left=1.6in,
right=1.6in]{geometry}

\usepackage[varg]{txfonts} 
\renewcommand{\mathbb}{\varmathbb}

\ifnum\showkeys=1
\usepackage[color]{showkeys}
\fi

\ifnum\showcolorlinks=1
\usepackage[
pagebackref,
letterpaper=true,
colorlinks=true,
urlcolor=blue,
linkcolor=blue,
citecolor=OliveGreen,
]{hyperref}
\fi

\ifnum\showcolorlinks=0
\usepackage[
pagebackref,
letterpaper=true,
colorlinks=false,
pdfborder={0 0 0}
]{hyperref}
\fi

\usepackage{prettyref}

\newcommand{\savehyperref}[2]{\texorpdfstring{\hyperref[#1]{#2}}{#2}}

\newrefformat{eq}{\savehyperref{#1}{\textup{(\ref*{#1})}}}
\newrefformat{lem}{\savehyperref{#1}{Lemma~\ref*{#1}}}
\newrefformat{def}{\savehyperref{#1}{Definition~\ref*{#1}}}
\newrefformat{thm}{\savehyperref{#1}{Theorem~\ref*{#1}}}
\newrefformat{cor}{\savehyperref{#1}{Corollary~\ref*{#1}}}
\newrefformat{cha}{\savehyperref{#1}{Chapter~\ref*{#1}}}
\newrefformat{sec}{\savehyperref{#1}{Section~\ref*{#1}}}
\newrefformat{app}{\savehyperref{#1}{Appendix~\ref*{#1}}}
\newrefformat{tab}{\savehyperref{#1}{Table~\ref*{#1}}}
\newrefformat{fig}{\savehyperref{#1}{Figure~\ref*{#1}}}
\newrefformat{hyp}{\savehyperref{#1}{Hypothesis~\ref*{#1}}}
\newrefformat{alg}{\savehyperref{#1}{Algorithm~\ref*{#1}}}
\newrefformat{rem}{\savehyperref{#1}{Remark~\ref*{#1}}}
\newrefformat{item}{\savehyperref{#1}{Item~\ref*{#1}}}
\newrefformat{step}{\savehyperref{#1}{step~\ref*{#1}}}
\newrefformat{conj}{\savehyperref{#1}{Conjecture~\ref*{#1}}}
\newrefformat{fact}{\savehyperref{#1}{Fact~\ref*{#1}}}
\newrefformat{prop}{\savehyperref{#1}{Proposition~\ref*{#1}}}
\newrefformat{prob}{\savehyperref{#1}{Problem~\ref*{#1}}}
\newrefformat{claim}{\savehyperref{#1}{Claim~\ref*{#1}}}
\newrefformat{relax}{\savehyperref{#1}{Relaxation~\ref*{#1}}}
\newrefformat{red}{\savehyperref{#1}{Reduction~\ref*{#1}}}
\newrefformat{part}{\savehyperref{#1}{Part~\ref*{#1}}}

\newcommand{\Sref}[1]{\hyperref[#1]{\S\ref*{#1}}}

\usepackage{nicefrac}

\let\nfrac=\nicefrac

\newcommand{\half}{\nicefrac12}

\ifnum\usemicrotype=1
\usepackage{microtype}
\fi

\ifnum\showauthornotes=1
\newcommand{\Authornote}[2]{{\sffamily\small\color{red}{[#1: #2]}}}
\newcommand{\Authorcomment}[2]{{\sffamily\small\color{gray}{[#1: #2]}}}
\newcommand{\Authorstartcomment}[1]{\sffamily\small\color{gray}[#1: }

\newcommand{\Authorfnote}[2]{\footnote{\color{red}{#1: #2}}}
\newcommand{\Authorfixme}[1]{\Authornote{#1}{\textbf{??}}}
\newcommand{\Authormarginmark}[1]{\marginpar{\textcolor{red}{\fbox{\Large #1:!}}}}
\else
\newcommand{\Authornote}[2]{}
\newcommand{\Authorcomment}[2]{}
\newcommand{\Authorstartcomment}[1]{}

\newcommand{\Authorfnote}[2]{}
\newcommand{\Authorfixme}[1]{}
\newcommand{\Authormarginmark}[1]{}
\fi

\newcommand{\Dnote}{\Authornote{D}}

\newcommand{\Pnote}{\Authornote{P}}

\newcommand{\PGnote}{\Authornote{PG}}

\newcommand{\Bnote}{\Authornote{B}}

\newcommand{\Mnote}{\Authornote{M}}

\ifnum\showfixme=0

\fi

\usepackage{boxedminipage}

\newenvironment{mybox}
{\center \noindent\begin{boxedminipage}{1.0\linewidth}}
{\end{boxedminipage}
\noindent
}

\newcommand{\paren}[1]{(#1)}
\newcommand{\Paren}[1]{\left(#1\right)}

\newcommand{\abs}[1]{\lvert#1\rvert}
\newcommand{\Abs}[1]{\left\lvert#1\right\rvert}

\newcommand{\card}[1]{\lvert#1\rvert}

\newcommand{\set}[1]{\{#1\}}
\newcommand{\Set}[1]{\left\{#1\right\}}

\newcommand{\norm}[1]{\lVert#1\rVert}

\newcommand{\normt}[1]{\norm{#1}_2}

\newcommand{\snorm}[1]{\norm{#1}^2}

\newcommand{\iprod}[1]{\langle#1\rangle}

\newcommand{\Esymb}{\mathbb{E}}
\newcommand{\Psymb}{\mathbb{P}}
\newcommand{\Vsymb}{\mathbb{V}}

\DeclareMathOperator*{\E}{\Esymb}
\DeclareMathOperator*{\Var}{\Vsymb}
\DeclareMathOperator*{\ProbOp}{\Psymb}

\renewcommand{\Pr}{\ProbOp}

\newcommand{\Prob}[2][]{\Pr_{{#1}}\Set{#2}}

\newcommand{\textparen}[1]{\text{(#1)}}

\ifx\because\undefined
\newcommand{\because}[1]{\textparen{because #1}}
\else
\renewcommand{\because}[1]{\textparen{because #1}}
\fi

\newcommand{\bits}{\{0,1\}}
\newcommand{\sbits}{\{\pm1\}}

\newcommand{\super}[2]{#1^{\paren{#2}}}

\newcommand{\vbig}{\vphantom{\bigoplus}}

\newcommand{\sm}{\setminus}

\newcommand{\defeq}{\stackrel{\mathrm{def}}=}

\newcommand{\seteq}{\mathrel{\mathop:}=}

\newcommand{\from}{\colon}

\renewcommand{\vec}[1]{{\bm{#1}}}

\newcommand{\mper}{\,.}
\newcommand{\mcom}{\,,}

\newcommand\bdot\bullet

\newcommand{\Ind}{\mathbb I}

\renewcommand{\th}{\textsuperscript{th}\xspace}

\DeclareMathOperator{\Inf}{Inf}

\DeclareMathOperator{\poly}{poly}
\DeclareMathOperator{\qpoly}{qpoly}
\DeclareMathOperator{\qpolylog}{qpolylog}
\DeclareMathOperator{\qqpoly}{qqpoly}

\DeclareMathOperator{\polylog}{polylog}
\DeclareMathOperator{\supp}{supp}
\DeclareMathOperator{\dist}{dist}
\DeclareMathOperator{\sign}{sign}

\DeclareMathOperator{\tsum}{{\textstyle \sum}}

\newcommand{\diffmacro}[1]{\,\mathrm{d}#1}

\newcommand{\etal}{et al.\xspace}

\newcommand{\Hastad}{H{\aa}stad\xspace}

\newcommand{\N}{\mathbb N}
\newcommand{\R}{\mathbb R}

\newcommand{\Rnn}{\R_+}

\newcommand{\problemmacro}[1]{\texorpdfstring{\textsc{#1}}{#1}\xspace}

\newcommand{\uniquegames}{\problemmacro{Unique Games}}

\newcommand{\sparsestcut}{\problemmacro{Sparsest Cut}}

\newcommand{\smallsetexpansion}{\problemmacro{Small-Set Expansion}}

\newcommand{\cA}{\mathcal A}

\newcommand{\cC}{\mathcal C}
\newcommand{\cD}{\mathcal D}

\newcommand{\cG}{\mathcal G}
\newcommand{\cH}{\mathcal H}

\newcommand{\cN}{\mathcal N}
\newcommand{\cO}{\mathcal O}

\newcommand{\cS}{\mathcal S}
\newcommand{\cT}{\mathcal T}

\newcommand{\cV}{\mathcal V}

\renewcommand{\leq}{\leqslant}
\renewcommand{\le}{\leqslant}
\renewcommand{\geq}{\geqslant}
\renewcommand{\ge}{\geqslant}

\ifnum\showdraftbox=1

\else

\fi

\let\epsilon=\varepsilon

\numberwithin{equation}{section}

\newcommand{\MYstore}[2]{%
  \global\expandafter \def \csname MYMEMORY #1 \endcsname{#2}%
}

\newcommand{\MYload}[1]{%
  \csname MYMEMORY #1 \endcsname%
}

\newcommand{\MYnewlabel}[1]{%
  \newcommand\MYcurrentlabel{#1}%
  \MYoldlabel{#1}%
}

\newcommand{\MYdummylabel}[1]{}

\newcommand{\torestate}[1]{%
  \let\MYoldlabel\label%
  \let\label\MYnewlabel%
  #1%
  \MYstore{\MYcurrentlabel}{#1}%
  \let\label\MYoldlabel%
}

\newcommand{\restatetheorem}[1]{%
  \let\MYoldlabel\label
  \let\label\MYdummylabel
  \begin{theorem*}[Restatement of \prettyref{#1}]
    \MYload{#1}
  \end{theorem*}
  \let\label\MYoldlabel
}

\newcommand{\restatelemma}[1]{%
  \let\MYoldlabel\label
  \let\label\MYdummylabel
  \begin{lemma*}[Restatement of \prettyref{#1}]
    \MYload{#1}
  \end{lemma*}
  \let\label\MYoldlabel
}

\newcommand{\restateprop}[1]{%
  \let\MYoldlabel\label
  \let\label\MYdummylabel
  \begin{proposition*}[Restatement of \prettyref{#1}]
    \MYload{#1}
  \end{proposition*}
  \let\label\MYoldlabel
}

\newcommand{\restatefact}[1]{%
  \let\MYoldlabel\label
  \let\label\MYdummylabel
  \begin{fact*}[Restatement of \prettyref{#1}]
    \MYload{#1}
  \end{fact*}
  \let\label\MYoldlabel
}

\newcommand{\restate}[1]{%
  \let\MYoldlabel\label
  \let\label\MYdummylabel
  \MYload{#1}
  \let\label\MYoldlabel
}

\newcommand{\addreferencesection}{
  \phantomsection
  \addcontentsline{toc}{section}{References}
}

\newcommand{\sse}{\subseteq}

\newcommand{\e}{\epsilon}
\newcommand{\eps}{\epsilon}
\newcommand{\eset}{\emptyset}

\let\origparagraph\paragraph
\renewcommand{\paragraph}[1]{\origparagraph{#1.}}

\allowdisplaybreaks

\newcommand{\cclassmacro}[1]{\texorpdfstring{\textbf{#1}}{#1}\xspace}

\newcommand{\np}{\cclassmacro{NP}}

\newcommand{\RM}{\ensuremath{\text{\sf RM}}}

\let\pref=\prettyref

\DeclareMathOperator{\cond}{\Phi}

\newcommand{\sst}{\substack}
\newcommand{\GF}[1]{\mathbb F_{#1}}

\DeclareMathOperator{\TC}{Test_{\mc{C}}}

\newcommand{\CC}{\mathcal{C}}
\newcommand{\vc}[1]{(-1)^{#1}}
\newcommand{\mb}[1]{#1} 
\newcommand{\mc}[1]{\mathcal{#1}}

\newcommand{\wt}{{\sf wt}}
\newcommand{\eat}[1]{}

\newcommand{\dashmaxtwolin}{\problemmacro{-Max-2Lin}}

\newcommand{\LH}{\ensuremath{\mathsf{LH}}\xspace}
\newcommand{\SA}{\ensuremath{\mathsf{SA}}\xspace}

\newcommand{\NS}{\mathsf{NS}}

\title{Making the long code shorter, with applications to the Unique Games Conjecture}

\author{Boaz Barak\thanks{Microsoft Research New England, Cambridge
    MA.} \and Parikshit Gopalan\thanks{Microsoft Research-Silicon
    Valley.} \and Johan \Hastad\thanks{Royal Insitute of Technology,
    Stockholm, Sweden.} \and Raghu Meka\thanks{IAS, Princeton. Work
    done in part while visiting Microsoft Research, Silicon Valley.}
  \and Prasad Raghavendra\thanks{Georgia Institute of Technology,
    Atlanta, GA.} \and David Steurer\thanks{Microsoft Research New
    England, Cambridge MA.}}

\begin{document}

\maketitle
\thispagestyle{empty}

\begin{abstract}
The \emph{long code} is a central tool in hardness of approximation, especially in questions related to the unique games conjecture. We construct a new code that is exponentially more efficient, but can still be used in many of these applications. Using the new code we obtain exponential improvements over several known results, including the following:

\begin{enumerate}

\item For any $\e>0$, we show the existence of an $n$ vertex graph $G$ where every set of  $o(n)$ vertices has expansion $1-\e$, but $G$'s adjacency matrix has more than $\exp(\log^{\delta}n)$ eigenvalues larger than $1-\e$, where $\delta$ depends only on $\e$. This answers an open question of Arora, Barak and Steurer (FOCS 2010) who asked whether one can improve over the noise graph on the Boolean hypercube that has $\poly(\log n)$ such eigenvalues.

\item A gadget that reduces unique games instances with linear constraints modulo $K$ into instances with alphabet $k$ with a blowup of $K^{\polylog(K)}$, improving over the previously known gadget with blowup of $2^{\Omega(K)}$.

\item An $n$ variable integrality gap for Unique Games that that survives $\exp(\poly(\log\log n))$  rounds of the SDP + Sherali Adams hierarchy, improving on the previously known bound of $\poly(\log\log n)$.
\end{enumerate}

We show a connection between the local testability of linear
codes and small set expansion in certain related Cayley graphs, and
use this connection to derandomize the noise graph on the Boolean hypercube.
\end{abstract}

\clearpage

\ifnum\showtableofcontents=1
{
\tableofcontents
\thispagestyle{empty}
 }
\fi

\clearpage

\section{Introduction}

Khot's \emph{Unique Games Conjecture}~\cite{Khot02a} (UGC) has been the focus of intense research effort in the last few years. The conjecture posits the hardness of approximation for a certain constraint satisfaction problem, and shows promise to settle many open questions in theory of approximation algorithms. Specifically, an instance $\Gamma$ of the \uniquegames problem  with $n$ variables and alphabet $\Sigma$ is described by a collection of constraints of the form $(x,y,\pi)$ where $\pi$ is a permutation over $\Sigma$. An \emph{assignment} to $\Gamma$ is a mapping $f$ from $[n]$ to $\Sigma$, and $f$'s value is the fraction of constraints $(x,y,\pi)$ such that $f(y)=\pi(f(x))$. The Unique Games Conjecture is that for any $\e>0$, there is some finite $\Sigma$ such that it is \np hard to distinguish between the case that a \uniquegames instance $\Gamma$ with alphabet $\Gamma$ has an assignment satisfying $1-\e$ fraction of the constraints, and the case that every assignment satisfies at most $\e$ fraction of $\Gamma$'s constraint.

Many works have been devoted to studying the plausibility of the UGC, as well as exploring its implications and obtaining unconditional results motivated by this effort. Tantalizingly, at the moment we have very little evidence for the truth of this conjecture. One obvious reason to believe the UGC is that no algorithm is known to contradict it, though that of course may have more to do with our proof techniques for algorithm analysis than actual computational difficulty. Thus perhaps the strongest evidence for the conjecture comes from results showing particular instances  on which certain natural algorithms will fail to solve the problem. However, even those integrality gaps are quantitatively rather weak. For example, while Arora, Barak and Steurer~\cite{AroraBS10} showed a subexponential upper bound on an algorithm for the \uniquegames and the related \smallsetexpansion problem, the hardest known instances for their algorithm only required quasipolynomial time~\cite{Kolla10}. Similarly (and related to this), known integrality gaps for \uniquegames and related problems do not rule out their solution by an $O(\log n)$-round semidefinite hierarchy, an algorithm that can be implemented in quasipolynomial (or perhaps even polynomial~\cite{BarakRS11}) time.

The \emph{long code} has been a central tool in many of these works. This is the set of ``dictator'' functions mapping $\GF 2^N$ to $\GF 2$ that have the form $x_1\ldots x_N \mapsto x_i$ for some $i$.
Many hardness reductions (especially from \uniquegames) and constructions of integrality gap instances use the long code as a tool. However, this is also the source of their inefficiency, as the long code is indeed quite long. Specifically, it has only $N$ codewords but dimension $2^N$, which leads to exponential blowup in many of these applications.
  In this work, we introduce a different code, which we call the ``short code'',  that is exponentially more efficient, and can be used in the long code's place in many of these applications, leading to significant quantitative improvements. In particular, we use our code to show instances on which the \cite{AroraBS10} algorithm, as well as certain semidefinite hierarchies, require almost sub-exponential time, thus considerably strengthening the known evidence in support of the Unique Games Conjecture.  Moreover, our results open up possibilities for \emph{qualitative} improvements as well, in particular suggesting a new approach to prove the Unique Games Conjecture via an efficient alphabet reduction.

\subsection{Our results}

At the heart of the long code's applications lie its connection with the \emph{noisy Hypercube}. This is the weighted graph $H_{N,\e}$ whose vertices are elements in $\GF 2^N$ where a random neighbor of $x\in \GF 2^N$ is obtained by flipping each bit of $x$ independently with probability $\e$.\footnote{This graph is closely related and has similar properties to the unweighted graph where we connect $x$ and $y$ if their Hamming distance is at most $\e N$.} It is not too hard to show that the codewords of the long code correspond to the top eigenvectors of the noisy hypercube which also give the minimal bisections of the graph, cutting only an $\e$ fraction of edges. In addition, several converse results are known, showing that bisections (and more general functions) cutting few edges are close to these top eigenvectors (or {\sl dictatorships}) in some sense. (One such result is the ``Majority is Stablest'' Theorem of ~\cite{MosselOO05}.)  The inefficiency of the longcode is manifested in the fact that the number of vertices of the noisy cube is exponential in the number $N$ of its top eigenvectors.

\paragraph{The short code}  Another way to describe the long code is that it encodes $x\in\GF 2^n$ by a binary vector $v_x$ of length $2^{2^n}$ where $v_x(f)=f(x)$ for every function $f:\GF 2^n\rightarrow \GF 2$. This view also accounts for the name ``long code'', since one can see that this is the longest possible encoding of $x$ without having repeated coordinates. For every subset $\cD$ of functions mapping $\GF 2^n$ to $\GF 2$, we define the \emph{$\cD$-short code} to be the code that encodes $x$ by a vector $v_x$ of length $|\cD|$ where $v_x(f)=f(x)$ for every $f\in \cD$. Note that this is a very general definition that encapsulates any code without repeated coordinates. For $d \in \N$, we define the \emph{$d$-short code} to be the the $\cD$-short code where $\cD$ is the set of all polynomials over $\GF 2^n$ of degree at most $d$. Note that the $1$-short code is the Hadamard code, while the $n$-short code is the long code. We use the name ``short code'' to denote the $d$ short code for $d=O(1)$. Note that the short code has $2^n$ codewords and dimension roughly $2^{n^d}$, and hence only quasipolynomial blowup, as opposed to the exponential blowup of the long code.  Our main contribution is a construction of a ``derandomized''  noisy cube, which is a small subgraph of the noisy cube that enjoys the same relations to the short code (including a ``Majority is Stablest'' theorem) as the original noisy cube has to the long code. As a result, in many applications one can use the short code and the derandomized cube in place of the long code and the noisy cube, obtaining an exponential advantage. Using this approach we obtain the following results:

\paragraph{Small set expanders with many large eigenvalues} Our first application, and the motivation to this work, is a question of Arora, Barak and Steurer \cite{AroraBS10}: How many eigenvectors with eigenvalue at least $1-\e$ can an $n$-vertex \emph{small set expander} graph have? We say a graph is a small set expander (SSE) if all sufficiently small subsets of vertices have, say, at least $0.9$ fracton of their neighbors outside the set.
 \cite{AroraBS10}  showed an upper bound of $n^{O(\e)}$ on the number of large (i.e., greater than $1-\e$) eigenvalues of a small set expander. Arora et al.~then observed that the subspace enumeration algorithm of~\cite{KollaT07,Kolla10} for approximating small set expansion in an input graph takes time at most exponential in this number, which they then use to give an algorithm with similar running time for the \uniquegames problem. Up to this work, the best lower bound was $\polylog(n)$, with the example being the noisy cube, and hence as far as we knew the algorithm of \cite{AroraBS10} could solve the small set expansion problem in quasipolynomial time, which in turn might have had significant implications for the \uniquegames problem as well. \Mnote{Is the last phrase contentious?}
 Our derandomized noisy cube yields an example with an almost polynomial number of large eigenvalues:

\begin{itheorem}\label{ithm:sse} For every $\e>0$, there is an $n$-vertex small set expander graph with $2^{(\log n)^{\Omega(1)}}$ eigenvectors with corresponding eigenvalues at least $1-\e$.
\end{itheorem}

Theorem~\ref{ithm:sse} actually follows from a more general result
connecting locally testable codes to small set expanders, which we
instantiate with the Reed Muller code. See Section~\ref{sec:techniques} for details.

\paragraph{Efficient integrality gaps} There is a standard
semidefinite program (SDP) relaxation for the \uniquegames problem,
known as the ``basic SDP'' \cite{KhotV05,RaghavendraS09c}. \Bnote{add refernce to some place where
it's nicely described}  Several works have shown upper and lower
bounds on the approximation guarantees of this relaxation, and for
constant alphabet size, the relation between the alphabet size and
approximation guarantee is completely
understood~\cite{CharikarMM06}. However, for unbounded alphabet, there
was still a big gap in our understanding of the relation between the
approximation guarantee and the number of variables. Gupta and Talwar
\cite{GuptaT06} \Bnote{who?} showed that if the relaxation's value is $1-\e$, there is an assignment satisfying $1-O(\e \log n)$ fraction of constraints. On the other hand, Khot and Vishnoi~\cite{KhotV05} gave an integrality gap instance where the relaxation's value was $1-1/\poly(\log\log n)$\footnote{Throughout, for any function $f$, $\poly(f(n))$ denotes a function $g$ satisfying $g(n) = f(n)^{\Omega(1)}$.} but the objective value (maximum fraction of constraints satisfied by any assignment) was $o(1)$. It was a natural question whether this could be improved (e.g., see \cite{JamesLeeBlog}), and indeed our short code allows us to obtain an almost exponential improvement:

\begin{itheorem}\label{ithm:efficient-gap}  There is an $n$-variable instance of \uniquegames with objective value $o(1)$ but for which the standard semidefinite programming (SDP) relaxation has value at least $1-1/\qpolylog(n)$.\footnote{For functions $f,g:\N\to [0,\infty)$ we write  $f = \qpoly(g)$ if $f=\exp(\polylog(g))$. That is, if there are constants $C>c>0$ such that for all sufficiently large $n$, $\exp((\log g(n))^c) \leq f(n) \leq \exp((\log g(n))^C)$. (Note that we allow $c<1$, and so $f = \qpoly(g)$ does not imply that $f>g$.) Similarly, we define $\qpolylog(g)=\qpoly(\log g)$ and write $f =\qqpoly(g)$ if $f=\exp(\exp(\poly(\log\log g)))$. }
\end{itheorem}

\paragraph{Integrality gaps for SDP hierarchies} Our best evidence for
the hardness of the Unique Games Conjecture comes from integrality gap
instances for semidefinite programming \emph{hierarchies}. These are
strengthened versions  of the basic SDP where one obtains tighter
relaxations by augmenting them with additional constraints, we refer to~\cite{ChlamtacT10} for a good overview of SDP hierarchies. These hierarchies are generally paramaterized by a number $r$ (often called the \emph{number of rounds}), where the first round corresponds to the Basic SDP, and the $n^{th}$ round (where $n$ is the instance size) corresponds to the exponential brute force algorithm that always computes an optimal answer. Generally,  the $r^{th}$-round of each such hierarchy can be evaluated in $n^{O(r)}$ time (though in some cases $n^{O(1)}2^{O(r)}$ time suffices~\cite{BarakRS11}). In this paper we consider two versions of these hierarchies--- the \SA hierarchy and the weaker \LH hierarchy. Loosely speaking, the $r^{th}$ round of the \SA hierarchy adds the constraints of the $r^{th}$ round of the Sherali-Adams linear programming hierarchy (see~\cite{SheraliA90}) to the Basic SDP; the $r^{th}$ round of the \LH hierarchy augments the Basic SDP with the constraints that and subset of $r$ vectors from the vector solutions embeds isometrically into the $\ell_1$ metric. (See Section~\ref{sec:hierarchy} and \cite{RaghavendraS09c} for more details.)

Barak, Raghavendra and Steurer~\cite{BarakRS11} (see also~\cite{GuruswamiS11}) showed that for every $\e>0$, $n^{\e}$ rounds of the \SA hierarchy yields a non-trivial improvement over the basic SDP . The unique games conjecture predicts that this is optimal, in the sense that  $n^{o(1)}$ rounds of any hierarchy should not improve the worst-case approximation ratio above the basic SDP.\footnote{This is under the widely believed assumption that $\mathbf{NP} \nsubseteq \mathbf{Dtime}(\exp(n^{o(1)})$.} 
 However, this prediction is far from being verified, with the best lower bounds given by \cite{RaghavendraS09c} (see also~\cite{KhotS09}) who showed instances that require $\log^{\Omega(1)} n$ rounds for the \LH hierarchy, and $(\log \log n)^{\Omega(1)}$ rounds for the \SA hierarchy. Moreover, these instances are \emph{known} to be solvable in quasipolynomial time~\cite{Kolla10} and in fact via $\polylog(n)$ rounds of the \SA hierarchy~\cite{BarakRS11} \Mnote{Just to confirm, it is \SA and not (approximate) Lasserre hierarchy in BRS?}. Thus prior work gave no evidence that the unique games problem cannot be solved in quasipolynomial time. In this work we obtain almost-exponentially more efficient integrality gaps, resisting  $\qpoly(\log n)$ rounds of the $\SA$ hierarchy and $\qqpoly(n)$ rounds of the $\LH$ hierarchy. The latter is the first superlogarithmic SDP hierarchy lower bound for \uniquegames for any SDP hierarchy considered in the literature.

\begin{itheorem}\label{ithm:hierarhcy-gap}  For every $\e>0$ there is some $k=k(\e)$, such that for every $n$ there is an $n$ variable instance $\Gamma$ of \uniquegames with alphabet size $k$
such that the objective value of $\Gamma$ is at most $\e$, but the value on $\Gamma$ of both $\qpoly(\log n)$ rounds of the $\SA$ hierarchy  and $\qqpoly(n)$ rounds of the $\LH$ hierarchy is at least $1-\e$.
\end{itheorem}

A corollary of the above theorem is a construction of an $n$-point metric of negative type such that all sets of size up to some $k = \qqpoly(n)$ embed isometrically into $\ell_1$ but the whole metric requires $\qpolylog(n)$ distortion to embed into $\ell_1$. We remark that Theorem~\ref{ithm:hierarhcy-gap}  actually yields a stronger result than stated here--- as a function of $k$, our results (as was the case with the previous ones) obtain close to optimal gap between the objective value and the SDP value of these hierarchies; in particular we show that in the above number of rounds one cannot improve on the approximation factor of the Geomans-Williamson algorithm for Max Cut. It is a fascinating open question whether these results can be extended to the stronger \emph{Lasserre} hierarchy. Some very recent results of Barak, Harrow, Kelner, Steurer and Zhou~\cite{BarakHKSZ11} (obtained subsequent to this work), indicate that new ideas may be needed to do this, since the \uniquegames instances constructed here and in prior works are not integrality gaps for some absolute constant rounds of the Lasserre hierarchy.

%

\paragraph{Alphabet reduction gadget} Khot, Kindler, Mossel and O'Donnel~\cite{KhotKMO04} used the long code to show an ``alphabet reduction'' gadget for unique games. They show how to reduce a unique game instance with some large alphabet $K$ to an instance with arbitrarily small alphabet. (In particular, they showed how one can reduce arbitrary unique games instances into binary alphabet instances, which turns out to be equivalent to the \emph{Max Cut} problem.) However, quantitatively their result was rather inefficient, incurring an exponential in $K$ blowup of the instance. By replacing the long code with our ``short code'', we obtain a more efficient gadget, incurring only a \emph{qusipolynomial} blowup. One caveat is that, because the short code doesn't support arbitrary permutations, this reduction only works for unique games instances whose constraints are affine functions over $\GF 2^k$ where $k=\log K$; however this class of unique games seems sufficiently rich for many applications.\footnote{For example, because the multiplicative group of the field $\GF{2^n}$ is cyclic, one can represent constraints of the form $x_i-x_j = c_{i,j} \pmod{2^n-1}$ as linear constraints over $\GF 2^n$ (i.e., constraints of the form $x_i = C_{i,j}x_j$ where $C_{i,j}$ is an invertible linear map over $\GF 2^n$).}

\begin{itheorem}\label{ithm:alph-red} For every $\e$ there are $k,\delta$, and a reduction that for every $\ell$ maps any $n$-variable \uniquegames  instance  $\Gamma$ whose constraints are affine permutations over alphabet $\GF 2^{\ell}$ into an $n\cdot \exp(\poly(\ell,k))$-variable \uniquegames instance $\Gamma'$ of alphabet $k$, such that if the objective value of $\Gamma$ is larger than $1-\delta$, then the objective value of $\Gamma'$ is larger than $1-\e$, and if the objective value of $\Gamma$ is smaller than $\delta$, then the objective value of $\Gamma'$ is smaller than $\e$.
\end{itheorem}

Once again, our quantitative results are stronger than stated here, and as in~\cite{KhotKMO04}, we obtain nearly optimal relation between the alphabet size $k$ and the soundness and completeness thresholds. In particular for $k=2$ our results match the parameters of the Max Cut algorithm of Geomans and Williamson. Our alphabet reduction gadget suggests a new approach to proving the unique games conjecture by using it as an ``inner PCP''. For example, one could first show hardness of unique games with very large alphabet (polynomial or even subexponential in the number of variables) and then applying alphabet reduction. At the very least, coming up with plausible hard instances for unique games should be easier with a large alphabet.

\begin{remark}The long code is also used as a tool in applications that do not involve the unique games conjecture. On a high level, there are two properties that make the long code useful in hardness of approximation: \textbf{(i)} it has a $2$ query test obtained from the noisy hypercube and \textbf{(ii)} it has many symmetries, and in particular one can read off any function of $x$ from the $x^{th}$ codeword. Our short code preserves property \textbf{(i)} but (as is necessary for a more efficient code) does not preserve property \textbf{(ii)}, as one can only read off low degree polynomials of $x$ (also it is only symmetric under affine transformations). We note that if one does not care about property \textbf{(i)} and is happy with a $3$ query test, then it's often possible to use the Hadamard code which is more efficient than the short code (indeed it's essentially equal to the $d$-short code for $d=1$). Thus, at least in the context of hardness of approximation, it seems that the applications the short code will be most useful are those where property \textbf{(i)} is the crucial one.

Despite the name ``short code'', our code is not the shortest possible code. While in our applications, dimension linear in the number of codewords is necessary (e.g., one can't have a graph with more eigenvalues than vertices), it's not clear that the dimension needs to be polynomial. It is a very interesting open question to find shorter codes that can still be used in the above applications.
\end{remark}

\section{Our techniques} \label{sec:techniques}

To explain our techniques we focus on our first application--- the construction of a small set expander with many eigenvalues close to $1$. The best way to  view this construction
is as a derandomization of the noisy hypercube, and so it will be useful to recall why the noisy hypercube itself is a small set expander.

Recall that the $\e$-noisy hypercube is the graph $H_{N,\e}$ whose vertex set is $\{ \pm 1\}^N$ where we sample a neighbor of $x$ by flipping each
bit independently with probability $\e$. The eigenvectors in $H_{N,\e}$
are given by the parity functions $\chi_\alpha(x) = \prod_{i \in \alpha}x_i$ for subsets $\alpha
\subseteq [N]$ and the corresponding eigenvalues are $\lambda_\alpha =
(1-2\eps)^{|\alpha|}$. Thus $\lambda_\alpha$ only depends on the degree
$|\alpha|$ of $\chi_\alpha$. In particular, the ``dictator'' functions
$\chi_{\{i\}}(x) = x_i$ have eigenvalue $1 -2 \eps$ and  they
correspond to balanced cuts (where vertices are partitioned based on the value of $x_i$) with edge expansion $\eps$. As
$\alpha$ increases, $\lambda_\alpha$ decreases, becoming a constant
around $|\alpha| = O(1/\eps)$.

Given $f: \{\pm 1 \}^N \rightarrow \{0,1\}$ which is the indicator of
a set $S$, its Fourier expansion $f(x) =
\sum_{\alpha}\hat{f}(\alpha)\chi_\alpha(x)$ can be viewed as expressing
the vector $f$ in the eigenvector basis. The edge expansion of $S$ is determined
by the distribution of its Fourier mass; sets where most of the Fourier
mass is on large sets will expand well.
Given this connection, small-set expansion follows from the fact that the indicator functions of small
sets have most of their mass concentrated on large Fourier
coefficients. More precisely a set $S$ of measure $\mu$ has most of
its Fourier mass on coefficients of degree $\Omega(\log(1/\mu))$. This follows from the so-called
(2,4)-hypercontractive inequality for low-degree polynomials--- that for every degree $d$ polynomial $f$,
\begin{equation}
\E_{x\in \{\pm 1\}^N}[f(x)^4] \leq C \E_{x \in \{ \pm 1 \}^N}[ f(x)^2]^2 \label{eq:hyper-cont-cube}
\end{equation}
  for some $C$ depending only on $d$. (See  Section~\ref{sec:subspace-hyper} for the proof, though some intuition can be obtained by noting
  that if $f$ is a characteristic function of a set $S$ of measure $\mu=o(1)$ then $\E[ f^2]^2 = \mu^2$ and $\E[f^4]=\mu$ and hence Equation~(\ref{eq:hyper-cont-cube}) shows that
  $f$ cannot be an $O(1)$-degree polynomial.)


By a ``derandomized hypercube'' we mean a graph on much fewer vertices that still (approximately) preserves the above properties of the noisy hypercube.
Specifically we want to find a very small subset $\cD$ of $\{ \pm 1\}^N$ and a subgraph $G$ of $H_{N,\e}$ whose vertex set is $\cD$ such
that \textbf{(i)} $G$ will have similar eigenvalue profile to $H_{N,\e}$, and in particular have $N$ eigenvalues close to $1$
and \textbf{(ii)} $G$ will be a a small set expander. To get the parameters we are looking for, we'll need to have the size of $\cD$ be at most
$\qpoly(N)$.

A natural candidate is to take $\cD$ to be a random set, but it is not hard to show that this will not work. A better candidate might be a linear subspace $\cD \subseteq \GF 2^N$ that looks suitably pseudorandom. We show that in fact it suffices to choose a subspace $\cD$  whose dual $\cC = \cD^{\perp}$ is a sufficiently good locally testable code. (We identify $\GF 2^N$ with $\{ \pm 1\}^N$ via the usual map $(b_1,\ldots,b_N) \mapsto ((-1)^{b_1}, \ldots,(-1)^{b_N})$.)

Our construction requires an asymptotic family of $[N,K,D]_2$ linear codes $\cC \subseteq \GF
2^N$ where the distance $D$ tends to infinity. The code should have a
$\e N$-query local tester which when given a received word $\alpha \in
\GF 2^N$ samples a codeword $q$ of weight at most $\eps N$ from a
distribution $\cT$ on $\cC^\perp$ and accepts if
$\iprod{\alpha,q}=1$. The test clearly accepts codewords in
  $\cC$, we also require it to reject words that are distance at least
$D/10$ from every codeword in $\cC$ with probability $0.49$.
Given such a locally testable code $\cC$, we consider the
Cayley graph\footnote{Cayley graph are usually defined to be unweighted
  graph. However, the definition can be generalized straightforwardly
  to weighted graphs.} $G$ whose vertices are the codewords
of the dual code $\cD = \cC^{\perp}$ while the (appropriately weighted)
edges correspond to the distribution $\cT$. That is, a vertex of $G$ is a codeword $x \in \cD$, while a random
neighbor of $x$ is obtained by picking a random $q$ from $\cT$ and moving to $x+q$.

Because $\cD$ is a subspace, it is easy to show that the eigenvectors of $G$ are linear functions of
of the form $\chi_\alpha(x)$ for $x,\alpha \in \GF 2^N$ (where if $\alpha \oplus \alpha' \in \cC$ then $\chi_{\alpha}$ and $\chi_{\alpha'}$
are identical on $G$'s vertices). Moreover, from the way we designed the graph, for every $\alpha \in \GF 2^n$, the corresponding
eigenvalue $\lambda_\alpha$ is equal to $\E_{q\in \cT} [ (-1)^{\iprod{\alpha,q}} ] = 1  -2\Pr_\cT[\text{Test rejects } \alpha]$. This
connection between the spectrum of $G$ and the local testability of
$\cC$ allows us to invoke machinery from coding theory in our analysis.

From this one can deduce that the eigenvalue spectrum of $G$ does
indeed resemble the hypercube in the range close to $1$.
In particular each $\chi_{\{i\}}(x) = x_i$ is a distinct eigenvector
with eigenvalue $1 - 2\eps$, and gives a bad cut in $G$ (where
vertices are partitioned based on the value of $x_i$). On the other hand for any eigenvector
$\chi$ of $G$, choose $\alpha$ of minimal weight such that $\chi=\chi_{\alpha}$. Now if $|\alpha|>D/10$ this means
that the distance of $\alpha$ from $\cC$ is at least $D/10$, which
using the testing property implies that $\lambda_{\alpha} \leq 1 -
2\cdot 0.49 = 0.02$.  

If we can show that indicator functions of small sets have most of their Fourier mass
on such eigenvectors (with small eigenvalue), that will imply that small sets have good
expansion. For small subsets of the hypercube, recall that this is
proved using (2,4)-hpercontractivity for low-degree polynomials. 
The key observation is that the inequality
\begin{equation}
\E_{x\in \cD}[ f(x)^4 ] \leq C \E_{x\in\cD}[f(x)^2]^2 \label{eq:hyper-cont-derand}
\end{equation} still holds for all polynomials $f$ of degree $d <
D/4$. This is because the distance of $\cC$ is $D$, hence the distribution of a random $x$ in $\cD$ is $D$-wise independent, which means
that the expectation of any polynomial of degree at most $D$ is equal over such $x$ and over a uniform $x$ in $\{ \pm 1 \}^N$. Thus
(\ref{eq:hyper-cont-derand}) follows from (\ref{eq:hyper-cont-cube}), completing our proof.

We instantiate this approach with using for $\cC$ the Reed Muller code
consisting of polynomials in $n$ variables over $\GF 2$ of degree
$n-d-1$. This is a code of distance $D=2^{d-1}$. We note that the
degree $n -d - 1$ and hence the rate of the code $\cC$ are very high. The graph is
over the codewords of $\cD=\cC^{\perp}$ that is itself the Reed Muller
code of polynomials over $\GF 2^n$ of degree $d$. Our basic tester consists of selecting a random minimum weight
codeword of $\cD$.\footnote{For many applications we amplify the success of this tester by selecting a sum of $t$ random
such words, this corresponds to taking some power of the basic graph $\cG$ described.} Thus our graph $\cG$
has as its vertices the $d$ degree polynomials over $\GF 2^n$ with an edge
between every polynomials $p,q$ such that $p-q$ is a product of $d$
linearly-independent affine functions (as those are the minimal weight
codewords in the Reed Muller code). We use the optimal analysis of
Bhattacharyya, Kopparty, Schoenebeck, Sudan and
Zuckerman~\cite{BhattacharyyaKSSZ10} to argue about the local
testability of $\cC$ which is a high degree Reed Muller code. We should note that this test is very closely
related to the Gowers uniformity test that was
first analyzed in the work of Kaufman et al. \cite{AKKLR:05}, but our
application requires the stronger result from~\cite{BhattacharyyaKSSZ10}.

\subsection{Other applications} \label{sec:other}

We now briefly outline how we use the above tools to obtain more efficient versions of several other constructions such as alphabet reduction gadgets and integrality gaps for unique games and other problems.

\paragraph{Efficient integrality gaps for Unique Games} To beign with, the graph we construct can be used to prove Theorem~\ref{ithm:efficient-gap}. That is, a construction of an $M$ variable instance $\Gamma$ of unique games where every assignment can satisfy at most a very small (say $1/100$) fraction of the constraints, but for which the standard semidefinite programming (SDP) relaxation has value of at least $1-1/\qpoly(\log M)$.  The basic idea is to simply take the graph $\cG$ we constructed above, and turn it into an instance of unique games by considering it to be the \emph{label extended graph} of some unique games instance. We now elaborate a bit below, leaving the full details to Section~\ref{sec:effug}. Recall that a \uniquegames instance $\Gamma$ with $M$ variables and alphabet $\Sigma$ is described by a collection of constraints of the form $(x,y,\pi)$ where $\pi$ is a permutation over $\Sigma$. An \emph{assignment} to $\Gamma$ is a mapping $f$ from $[M]$ to $\Sigma$, and $f$'s value is the fraction of constraints $(x,y,\pi)$ such that $f(y)=\pi(f(x))$.  The \emph{label extended graph} corresponding to $\Gamma$ is the graph $G_{\Gamma}$ over vertices $[M]\times \Sigma$ where for every constraint of the form $(x,y,\pi)$ and $\sigma\in \Sigma$ we add an edge between $(x,\sigma)$ and $(y,\pi(\sigma))$. It is not hard to see that an assignment of value $1-\e$ corresponds to a subset $S$ containing exactly $M$ of $G_{\Gamma}$'s vertices with small expansion (i.e., $\e$ fraction of the edges from $S$ leave the set). Thus if $G_{\Gamma}$ is an expander for sets of measure $1/|\Sigma|$ in $G_{\Gamma}$  then there is no nearly satisfying assignment for the unique games instance $\Gamma$.  In our case, our graph $\cG$ has the degree $d$ polynomials over $\GF 2^n$ as its vertices, and we transform it into a unique game instance whose variables correspond to degree $d$ polynomials \emph{without linear terms}. The alphabet $\Sigma$ consists of all linear functions over $\GF 2^n$. We ensure that the graph $\cG$ is the label extended graph of $\Gamma$ by setting the permutations accordingly: given a polynomial $p$ without a linear term, and a function $q$ that is a product of $d$ affine functions,\footnote{Actually, to get better parameters, we take some power $t$ of $\cG$, meaning that we consider $q$ that is a sum of $t$ functions that are products of $d$ affine functions.} if we write $q=q'+q''$  where $q''$ is the linear part of $q$, then we add a constraint of the form $(p,p+q',\pi)$ where $\pi$ is the permutation that maps a linear function $r$ into $r+q''$. Some not too difficult calculations show that the top eigenvectors of our graph $\cG$ yield a solution for the semidefinite program for $\Gamma$ (if the top eigenvectors are $f^1,\ldots,f^K$, our vector solution will associate with each vertex $x$ the vector $(f^1(x),\ldots,f^K(x)$). By choosing carefully the parameters of the graph $\cG$, the instance $\Gamma$ will have SDP value $1-1/\qpoly(\log M)$ where $M$ is the number of variables.

\paragraph{Derandomized Invariance Principle} While hypercontractivity of low degree polynomials suffices for some applications of the long code, other applications require other theorems, and in particular the \emph{invariance principle}, shown for
the hypercube by Mossel, O'Donnel and Oleszkiewicz~\cite{MosselOO05}.Roughly speaking their invariance principle says that for ``nice'' functions $f$ on the vertices of the $N$-dimensional noisy hypercube, the distribution of $f(x)$
where $x$ is a random vertex is close to the distribution of $f(y)$
where $y$ consists of $N$ independent standard Gaussian random
variables (appropriately extending $f$ to act on $\R^N$). To obtain
more efficient version of these applications, we first show that the
same holds even when $x$ is a random vertex in our smaller subset of
$N$-dimensional strings -- the Reed--Muller codewords. Our central
tool is a  recent result by Meka and Zuckerman~\cite{MekaZ10} which
derandomizes the invariance principle of Mossel et al. Our key insight
 is that taking a random Reed--Muller codeword can in fact be viewed
 as an instantiation of the Meka-Zuckerman generator, which involves splitting the input into blocks via a pairwise independent hash function, and using independent $k$-wise independent distributions in each block. This allows us to obtain a version of the ``Majority is Stablest'' theorem for our graph, which is the main corollary of the invariance principle that is used in applications of the longcode. See Section~\ref{sec:stablestcodes} for more details.

\paragraph{Efficient alphabet reduction } With the ``Majority of Stablest'' theorem in hand, proving Theorem~\ref{ithm:alph-red} (efficient alphabet reduction for unique games), is fairly straightforward. The idea is to simply replace the noisy hypercube gadget used by~\cite{KhotKMO04} with our derandomized hypercube. This is essentially immediate in the case of alphabet reduction to binary alphabet (i.e., reduction to Max Cut) but requires a bit more work when reducing to a larger alphabet. See Section~\ref{sec:alphabetreduction} for more details.

\paragraph{Efficient hierarchy integrality gaps} Our proof Theorem~\ref{ithm:hierarhcy-gap} again works by plugging in our short code / derandomized noisy hypercube in place of the long code in the previous integrality gap constructions~\cite{KhotV05,KhotS09,RaghavendraS09c}. Specifically, these constructions worked by starting with an integrality gap for unique games where the basic SDP yields $1-1/r$, and then composing it with an alphabet reduction gadget to obtain a new instance; Raghavendra and Steurer~\cite{RaghavendraS09c} showed that the composed instances resist  $\poly(r)$ rounds of the $\SA$ hierarchy and $\exp(\poly(r))$ rounds of the $\LH$ hierarchy. These constructions used the noisy cube twice--- both to obtain the basic unique games gap instance, and to obtain the alphabet reduction gadget. We simply plug in our short code in both usages--- using for the basic unique games instance the efficient version obtained in Theorem~\ref{ithm:efficient-gap}, and for the alphabet reduction gadget the efficient version obtained in Theorem~\ref{ithm:alph-red}. (Luckily, our unique games instance has affine constraints and so is compatible with our alphabet reduction gadget.) The result essentially follows in a blackbox way from the analysis of \cite{RaghavendraS09c}. See Section~\ref{sec:hierarchy} for details.

\section{Preliminaries}

\Dnote{we should probably consistently say eigenfunction instead of
  eigenvector in the technical sections.
  The reason is that we refer to dictators or other characters also as
  functions.
  Also it is sometimes confusing for people to use expectation norms for
  vectors.
}

Let $G$ be a regular graph with vertex set $V$.
For a subset $S\subseteq V$ we define the \emph{volume} of $S$, denoted
$\mu(S)$, to be $|S|/|V|$.
We define the \emph{expansion} of $S$, denoted $\cond(S)$, to be the
probability over a random edge $(u,v)$, conditioned on $u\in S$ that
$v\not\in S$.
Equivalently (since $G$ is regular), $\cond(S) = G(S, V \sm S)/(\deg_G
\card S)$ where $\deg_G$ is the degree of the graph $G$ and $G(S,V\sm S)$ is
the number of edges going from $S$ to $V\sm S$. Throughout, we denote the normalized adjacency matrix of a graph $G$ also by $G$, and refer to the spectrum of the adjacency matrix as the spectrum of the graph $G$. Note that by definition, every regular graph has maximum eigenvalue $1$.
In this paper, we use \emph{expectation norms} for real-valued functions.
That is, for a function $f\from S\to\R$ and $p\geq 1$, we let $\norm{f}_p
\seteq (\E_{x\in S} |f(x)|^p)^{1/p}$.

Many of the unique games instances that appear in this work belong to
a special subclass of unique games, namely $\GF2^n\dashmaxtwolin$
instances defined below.
\begin{definition}
	Given a group $\mc{H}$, an $\mc{H}\dashmaxtwolin$ instance consists of a
	system of linear equations over the group $\mc{H}$ where each
	equation is of the form $x_{i} - x_{j} = c_{ij}$ for some
	$c_{ij} \in \mc{H}$.
\end{definition}

\paragraph{Locally Testable Codes}
Let $\cC$ be an $[N,K,D]_2$ code, that is, $\cC$ is a $K$-dimensional
linear subspace of $\GF2^N$ with minimum distance $D$ ($= \min \{\wt(x): x \in \cC\}$).
%
%
%
(In this paper, we are mostly interested in the extremely high rate regime
when $H = N - K$ is very small compared to $N$ and are happy with $D$
being some large constant.)
Let $\Delta(\mb{x},\mb{y}) \in \{0,\ldots,N\}$ denote Hamming distance
between $\mb{x},\mb{y} \in \GF2^N$.
For $\alpha \in \GF2^N$ and a code $\cC$ we define
\begin{displaymath}
  \Delta(\alpha,\mc{C}) \defeq \min_{\mb{c} \in  \mc{C}}\Delta(\alpha,\mb{c}).
\end{displaymath}

\begin{definition} \label{def:canonical-tester}
  %
  We say a distribution $\cT$ over $\GF2^N$ is a \emph{canonical tester}
  for $\cC$ if every vector in the support of the distribution $\cT$ is a
  codeword $q\in C^\bot$.
  The \emph{query complexity} of $\cT$ is the maximum weight of a vector in
  its support.
  The tester's \emph{soundness curve} $s_\cT\from \N\to [0,1]$ is defined as
  \begin{displaymath}
    s_\cT(k) \defeq \min_{\substack{\alpha\in \GF2^N\\\Delta(\alpha,\cC)\ge k}}
    \Prob[q\sim \cT]{\iprod{\alpha,q}=1}
    \mper
  \end{displaymath}
 \Dnote{what is the right definition? better to minimize only over vectors
    with distance \emph{equal} to $k$? The possible non-monotonicity of the
    rejection probability (as a function of the distance) is in some sense
    exactly the problem we are interested in.}
  Similarly, we denote the \emph{rejection probability} of $\cT$ for a
  vector $\alpha \in \GF 2^N$ by $s_\cT(\alpha)=\Prob[q\sim\cT]{\iprod{\alpha,q}=1}$.
  %
  %
  %
  %
  We let the \emph{query probability} $\tau\in [0,1]$ of a tester be the
  expected fraction of queried coordinates, that is, $\tau = \E_{q\sim \cT}
  \wt(q)/N$.
  We say that a tester $\cT$ with query probability $\tau$ is \emph{smooth}
  if for any coordinate $i\in [N]$, $\Prob[q\sim \cT]{q_i=1}=\tau$ and we
  say it is \emph{$2$-smooth} if in addition, for any two distinct coordinates
  $i\neq j$, $\Prob[q\sim \cT]{q_i=q_j=1}=\tau^2$.
  \Dnote{probably there is a shorter way to define smooth and
    $2$-smooth}
\eat{  A \emph{canonical tester} $\TC$ with query complexity $t$ for
  an $[N,K,D]_2$ code $\mc{C}$ picks a codeword $\mb{q} \in \mc{C}^\perp$ with $\wt(\mb{q})
  \leq t$ according to some distribution $\mc{T}$. It accepts $r \in \GF2^N$
  if $\mb{r}\cdot \mb{q} = 0$. For a monotone function $s:[0,1] \rightarrow [0,1]$, we say that the tester has soundness $s$ if it rejects
  every $\mb{r}$ that such that $\Delta(\mb{r},\mc{C}) \geq \eta D$ with
  probability at least $s(\eta)$. We say it is a $p$-\emph{good tester}
  if $s(1/10) > p$.
}
\end{definition}

If the tester $\cT$ is clear from the context, we will sometimes drop the
subscript of the soundness curve / rejection probability $s_\cT$.
In the setting of this paper, we will consider testers with query
probability slowly going to $0$ (with $N$). Further, given a canonical tester $\mc{T}$, it is easy to amplify the
probability of rejection by repeating the test and taking the XOR of the
results.
\eat{
\PGnote{I think we dont use this lemma, since we are doing the coninuous time random walk.
Deleted:
Formally, let $\cT^{\oplus \ell}$ denote the canonical tester that samples
$\super {q}1,\super{q} 2,\ldots,\super {q} \ell$ independently from $\mc{T}$ and
accepts $r \in \GF2^N$ if $\iprod{r,\sum_{i \in [\ell]} \super{q} i} = 0$.

\begin{lemma}
  \label{lem:xor}
  For every vector $r\in \GF2^N$,
  \begin{displaymath}
    s_{\cT^{\oplus \ell}}(r) = \tfrac12 (1-(1-2s_\cT(r))^\ell)\mper
  \end{displaymath}
  In particular, the soundness curve of the tester $\cT^{\oplus \ell}$
  satisfies $s_{\cT^{\oplus \ell}}(k)\ge \frac{1 - (1-2s_\cT(k))^\ell}{2}$.
\end{lemma}
\begin{proof}
  The value $\sum_{i \in [\ell]} \mb{r} \cdot \mb{q}_i$ is the XOR of
  $\ell$ independent random variables, each equalling  $1$ with
  probability at least $s_\cT(r)$.
\end{proof}}}

Finally, the following simple lemma gives some estimates for rejection probabilities of vectors for smooth testers.

\begin{lemma}
  \label{lem:smooth}
  If $\cT$ is a smooth canonical tester with query probability $\tau$, then
  $s_\cT(\alpha)\le \Delta(\alpha,\cC)\cdot \tau$ for every vector $\alpha\in\GF2^N$.
  Furthermore, if $\cT$ is $2$-smooth, then $s_\cT(\alpha)\ge
  (1-\gamma)\cdot \Delta(\alpha,\cC) \cdot \tau$ for every vector $\alpha\in\GF2^N$ with
  $\Delta(\alpha,\cC) \tau \le \gamma$.
\end{lemma}

\begin{proof}
Fix $\alpha \in \GF2^N$ and let $k=\Delta(\alpha,\cC)$.
  Without loss of generality, we may assume $\wt(\alpha)=k$.
  By renaming coordinates, we may assume $\alpha_1=\ldots = \alpha_k=1$ and
  $\alpha_{k+1}=\ldots= \alpha_N = 0$.
  Then, $s_\cT(\alpha)\le \Prob[q\sim\cT]{q_1=1} + \ldots +
  \Prob[q\sim\cT]{q_k=1} = k\cdot \tau$.
  On the other hand,
  \begin{displaymath}
    s_\cT(\alpha)
    \ge \sum_{i=1}^k \Prob[q\sim \cT]{q_i=1}
    - \sum_{0\le i<j\le k} \Prob[q\sim \cT]{q_i=q_j=1}
    \ge k\tau - k^2\tau^2
    \ge (1-\gamma)\cdot k\tau\mper  \qedhere
  \end{displaymath}

\end{proof}

We review the prerequisites for Majority is Stablest and Unique Games related results in the corresponding sections.


\section{Small Set Expanders from Locally Testable Codes}
In this section we first use some known properties of hypercontractive norms to give a sufficient condition for graphs to be small set expanders. We then describe a generic way to construct graphs satisfying this condition from locally testable codes, proving Theorem \ref{ithm:sse}.

\subsection{Subspace hypercontractivity and small set expansion} \label{sec:subspace-hyper}
Let $\cV$ be a subspace of the set of functions from $V$ to $\R$ for some finite set $V$. We denote by $P_{\cV}$ the projection operator to the space $\cV$. For $p,q  \geq 1$, we define
\[
\norm{\cV}_{p \to q} \defeq \max_{f:V\to\R} \tfrac{\norm{P_{\cV}f}_q}{\norm{f}_p} \mper
 \]
 We now relate this notion to small set expansion. We first show that  a subspace $\cV$ with bounded $(4/3)\to 2$ norm cannot contain the characteristic function of a small set:

\begin{lemma}\label{lem:34nosmallset}  Let $f:V\to \{0,1\}$ such that $\mu = \E_{x\in V}[f(x)]$ then $\normt{P_{\cV}f}^2 \leq \norm{\cV}_{4/3\to 2}^2\mu^{3/2}$.
\end{lemma}
\begin{proof} This is by direct calculation
\[
\normt{P_{\cV}f}^2 \leq  \norm{\cV}_{4/3\to 2}^2\norm{f}_{4/3}^2 = \norm{\cV}_{4/3\to 2}^2 \mu^{(3/4)\cdot 2}
\]
\end{proof}
Note that if $\norm{\cV}_{4/3\to 2}=O(1)$ and $\mu=o(1)$, then $\normt{P_{\cV}f}^2 = o(\normt{f}^2)$, meaning the projection of $f$ onto $V$ is small. It is often easier to work with the $2\to 4$ norm instead of the $4/3 \to 2$ norm. The following lemma allows us to use a bound on the former to bound the latter:

\begin{lemma}\label{lem:2to4imp43to2}
\[
\norm{\cV}_{4/3\to 2} \leq \norm{\cV}_{2\to 4}
\]
\end{lemma}
\begin{proof} Let $f:V \to \R$ and let $f' = P_Vf$. We know that
\begin{align*}
\E[f'^2] &=
\E[f'\cdot f] \quad \text{\;(since $f'$ is the projection of $f$)\;}
\\
&\le \E[ f'^4]^{1/4}\E[ f^{4/3} ]^{3/4} \quad\text{\;(by H\"older's inequality)}
\\
&= \E[(P_\cV f')^4]^{1/4} \E[f^{4/3}]^{3/4} \quad\text{\;(projection is idempotent)}\\
&\le \norm{\cV}_{2\to 4} \E[ (f')^2 ]^{1/2} \E[ f^{4/3} ]^{3/4}
\mper
\end{align*}
Dividing by $\normt{f} = \E[ f^2 ] ^{1/2}$ yields the result.
\end{proof}

We now conclude that graphs for which the top eigenspace has bounded
$2 \to 4$ norm are small set expanders. The lemma can be viewed
qualitatively as a generalization of one direction of the classical
Cheeger's inequality relating combinatorial expansion to eigenvalue
gap \cite{Cheeger70}.

\begin{lemma}\label{lem:hyper-to-sse} Let $G=(V,E)$ be  regular graph, and $\cV$ be the span of the eigenvectors of $G$ with eigenvalue larger than $\lambda$. Then, for every $S \subseteq V$,
\[
\cond(S) \geq 1-\lambda - \norm{\cV}_{2\to 4}^2\sqrt{\mu(S)}
\]
\end{lemma}
\begin{proof} Let $f$ be the characteristic function of $S$, and write $f = f' + f''$ where $f' = P_{\cV}f$ (and so $f''=f-f'$ is the projection to the eigenvectors with value at most $\lambda$). Let $\mu = \mu(S)$. We know that
\begin{equation}
\cond(S) = 1 - \iprod{f,Gf}/\normt{f}^2 = 1 - \iprod{f,Gf}/\mu \label{eq:expansioneval}
\end{equation}
By \pref{lem:34nosmallset}, and \pref{lem:2to4imp43to2},
\[
\iprod{f,Gf} =  \iprod{f',Gf'} + \iprod{f'',Gf''} \leq \normt{f'}^2 + \lambda\normt{f''}^2 \leq  \norm{\cV}_{4/3\to 2}^2\mu^{3/2} + \lambda\mu \leq \norm{\cV}_{2\to 4}^2 \mu^{3/2} + \lambda \mu\mper
\]
Plugging this into (\ref{eq:expansioneval}) yields the result.
\end{proof}

\subsection{Cayley graphs on codes} \label{sec:Cayley-graph-codes}

\newcommand{\Cay}{\mathrm{Cay}}
Motivated by the previous section, we now construct a graph for which the projection operator on to the top eigenspace is hypercontractive, i.e., has small $2\to4$ norm, while also having high rank.

Let $\cC\sse \GF2^N$ be an $[N,K,D]_2$ code. The graph we construct will be a Cayley graph with vertices indexed by $\cC^\bot$ and edges drawn according to a canonical local tester $\cT$ for $\cC$.
%
%
Let $\Cay(\cC^{\perp},\cT)$ denote the (weighted) Cayley graph with vertex set
$\mc{C}^\perp$ and edges generated by $\mc{T}$.
We describe the graph more precisely by specifying the neighbor
distribution for a random walk on the graph.
For a vertex $p\in \cC^\bot$, a random neighbor has the form $p+q$ with $q$
sampled from the tester $\cT$.
(Since the group $\cC^\bot$ has characteristic $2$, the graph
$\Cay(\cC^\bot,\cT)$ is symmetric for every tester $\cT$.)

%
We will argue that if the tester $\cT$ has small query complexity and good
soundness, then the graph $\Cay(\CC^{\perp},\cT)$ has many large
eigenvalues while being a small-set expander.

\begin{theorem}
\label{thm:main}
Let $\mc{C}$ be an $[N,K,D]_2$ linear code that has a canonical tester
$\cT$ with query complexity $\eps N$ and soundness curve $s()$ and let
$k<D/5$. The graph $\Cay(\mc{C}^\perp,\mc{T})$ has $2^{N-K}=2^H$ vertices with
at least $N/2$ eigenvalues larger than $1 - 4\e$.
All subsets $S$ of $\mc{C}^\perp$ have expansion at least
\[
\cond(S) \geq 2s(k) - 3^{k}\sqrt{\mu(S)}
\]
\end{theorem}

By Xoring the results of mulitple tests, one can let the
soundness $s(k)$ tend to $\half$.
Hence, if $s(k)$ is significantly larger than $\e$ (for appropriate $k$),
one can obtain a graph with many large eigenvalues such that small enough
sets have near-perfect expansion.

\Dnote{deleted: In particular, subsets $S$ of volume less than $\delta 9^{-D/5}$ have at least $2s(D/5)-\delta$ fraction of edges
leaving them.
Note that, for many codes and tester, as $D$ grows $s(D/5)$ tends to $1/2$.
}

\paragraph{Eigenfunctions and Eigenvalues}
We identify the graph $G = \Cay(\mc{C}^\perp,\mc{T})$ by its normalized adjacency matrix.
For every vector $\alpha \in \GF2^N$, the character $\chi_\alpha\from
\cC^\bot \to \sbits$ with $\chi_\alpha(\mb{p}) = (-1)^{\iprod{\alpha,
    \mb{p}}}$ is an eigenfunction of $G$.
If two vectors $\alpha, \beta\in \GF2^N$ belong to the same coset of
$\mc{C}$, they define the same character over $C^\bot$ since $\iprod{\alpha
  + \beta,\mb{p}} =0$ for all $p\in \cC^\bot$, while if $\alpha + \beta
\not\in \mc{C}$ then $\iprod{\chi_\alpha, \chi_\beta } =0$.
Thus, the set of characters of $\cC^\bot$ corresponds canonically to the
quotient space $\GF 2^N/\cC$.
If we fix a single representative $\alpha$ for every coset in
$\GF2^N/\mc{C}$, we have exactly $2^{N - K} = 2^H$ distinct, mutually
orthogonal characters. We define the degree of a character as follows:
\begin{equation}
\deg(\chi_\alpha)
= \min_{\mb{c} \in \mc{C}} \wt(\alpha + \mb{c})
= \Delta(\alpha,\mc{C})\mper
\end{equation}
Note that if $\deg(\chi_\alpha)<D/2$, then the minimum weight
representative in $\alpha+\cC$ is unique.
(This uniqueness will allow us later to define low-degree influences of
functions, see Section \ref{sec:stablestcodes}.)

We let $\lambda_\alpha$ denote the eigenvalue corresponding to character
$\chi_\alpha$. The following observation connects the soundness of the
canonical tester to the spectrum of $G$:

\begin{lemma}
\label{lem:eigenvalues}
For any $\alpha \in \GF2^N$, $\lambda_\alpha = 1-2s(\alpha)$.
\end{lemma}
\begin{proof}
From standard facts about Cayley graphs, it follows that
\begin{equation}
\label{eqn:test}
\lambda_\alpha = \E_{\mb{q} \in \mc{T}}[\chi_\alpha(\mb{q})] = \E_{\mb{q} \in
  \mc{T}}[(-1)^{\alpha\cdot \mb{q}}] = 1 - 2\Pr_{\mb{q} \in
  \mc{T}}[\alpha \cdot \mb{q} =1] = 1 - 2s(\alpha).
\end{equation}
\end{proof}

We use this to show that many {\em dictator cuts} in $G$ which correspond to characters with
degree $1$ have eigenvalues close to $1$. We let
$\lambda_i, \chi_i$ denote $\lambda_{\{i\}}, \chi_{\{i\}}$. As noted
before, for $D > 2$ these are distinct characters.

\begin{corollary}\label{cor:manyeigs}
We have $\lambda_i \geq 1 - 4\eps$ for at least $N/2$ coordinates $[i] \in N$.
\end{corollary}
\begin{proof}
We have $\lambda_i = 1 - 2\Pr_{\mb{q} \in \mc{T}}[\mb{q}_i =1]$.
Since $\wt(\mb{q}) \leq \eps N$ for every $\mb{q} \in \mc{T}$,
$$\sum_{i=1}^N\Pr_{\mb{q} \in \mc{T}}[\mb{q}_i =1]  \leq \eps N.$$
So we can have $\Pr_{\mb{q} \in \mc{T}}[\mb{q}_i =1]  \geq 2\eps$
for at most $N/2$ coordinates.
\end{proof}

Another immediate conseuqence of Lemma \ref{lem:eigenvalues} is that large degree characters have small eigenvalues.
\begin{corollary} \label{cor:spectrumbound}
  If $\deg(\chi_\alpha) \geq k$, then $\lambda_\alpha \leq 1 - 2s(k)$.
\end{corollary}
\PGnote{Deleted trivial proof.}
\eat{
\begin{proof}
We have $\lambda_\alpha = 1 - 2s(\alpha)$.
But since $\deg(\chi_\alpha) = \Delta(\alpha,\mc{C}) \geq
k$, the tester rejects with probability at least $s(k)$.
\end{proof}}


\paragraph{Subspace Hypercontractivity}



Given a function $f\from\mb{\mc{C}^\perp} \to \R$ we can write it
(uniquely) as a linear combination of the characters
$\set{\chi_\alpha}_{\alpha \in \GF2^N/\cC}$
\begin{displaymath}
  f(\mb{p})
  = \sum_{\alpha \in \GF2^N/\mc{C}}\hat{f}(\alpha)\chi_\alpha(\mb{p})\mcom
\end{displaymath}
where $\hat f(\alpha)=\iprod{\chi_\alpha,f}$ is the \emph{Fourier
  transform} of $f$ (over the abelian group $\cC^\bot$).

We define the \emph{degree} of $f$, denoted $\deg(f)$ to be $\max_{\alpha:\hat{f}(\alpha) \neq 0}
\deg(\chi_\alpha).$ Note that $\deg(f+g) \leq \max\{ \deg(f),\deg(g)\}$ and
$\deg(fg) \leq \deg(f) + \deg(g)$.
\eat{
\Bnote{Removed sentence:  "Because characters with degree $1$ do \emph{not} necessarily correspond to a dictatorship function
over $\GF 2^{\dim(\cC^{\perp})}$, our definition is not identical to the
usual definition of polynomial degree."
because Parikshit and David found it confusing. Note however that dictators do correspond to characters (= monomials) of degree $1$.}


%
\Dnote{last sentence sounds confusing to me.
  there is not really a canonical way to represent functions on $C^\bot$ as
  functions over $\GF2^{\dim(C^\bot)}$ (would have to fix a basis).
 }
}
The following crucial observation follows immediately from the fact that $\cC$ has minimum distance $D$.
\begin{fact}
The uniform distribution on $\mc{C}^\perp$ is $(D-1)$ wise
independent. That is, for any $\alpha \in \GF2^N$ such that $1 \leq \wt(\alpha) < D$ we have $\E_{\mb{p} \in
  \mc{C}^\perp}[\chi_\alpha(\mb{p})] = 0$.
\end{fact}

This fact has the following corollary:

\begin{lemma}
\label{lem:hc42}
Let $\ell < (D-1)/4$ and let $\cV$ be the subspace of functions with degree at most $\ell$. Then
$\norm{\cV}_{2\to 4} \leq 3^{\ell/2}$.
\end{lemma}
\begin{proof}
The proof follows from the following two facts:
\begin{enumerate}
\item This bound on the $2\to 4$ norm is known to hold for true low degree polynomials under the uniform distribution on the hypercube by the Bonami-Beckner-Gross inequality~\cite{OD:08}.
\item The expectation of polynomials of degree up to $4\ell < D-1$ are the same under the uniform distribution
and a $D-1$-wise independent distribution.
\end{enumerate}

Given $f:\R^n\to\R$, let $f^\ell$ denote its projection onto the space
$\cV$ spanned by characters where $\deg(\chi_\alpha) \leq \ell$. We have
\[
\norm{f^\ell}_4^4 = \E_{\mb{p}\in \cC^{\perp}}[ f^\ell(\mb{p})^4 ] = \E_{\mb{p} \in \{0,1\}^N}[f^\ell(\mb{p})^4 ] \mcom
\]
\[
\norm{f}_2^2 \geq \norm{f^\ell}_2^2 =  \E_{\mb{p}\in \cC^{\perp}}[ f^\ell(\mb{p})^2 ] = \E_{\mb{p} \in \{0,1\}^N}[f^\ell(\mb{p})^2 ] \mper
\]
By the $2 \to 4$ hypercontractivity for degree $\ell$ polynomials over $\{0,1\}^N$,
\[
\E_{\mb{p} \in \{0,1\}^N}[ f^\ell(\mb{p})^4 ] \leq 9^{\ell} \E_{\mb{p} \in \{0,1\}^N}[ f^\ell(\mb{p})^2]^2 \mper
\]
So we conclude that
\[
\E_{\mb{p} \in \cC^\perp}[ f^\ell(\mb{p})^4 ] \leq 9^{\ell} \E_{\mb{p}
  \in \cC^\perp}[ f^\ell(\mb{p})^2]^2 \leq 9^{\ell} \E_{\mb{p} \in
  \cC^\perp}[ f(\mb{p})^2]^2 \mcom
\]
which implies that $\norm{\cV}_{2\to 4} \leq 3^{\ell/2}$.
\end{proof}

Combining the above bound with Lemma  \ref{lem:hyper-to-sse} we get that, if the local tester rejects sufficiently far codewords with high probability, then the resulting graph is a small set expander:

\begin{corollary}
\label{cor:hc-sse}
  For every vertex subset $S$ in the graph $\Cay(\cC^\bot,\mc{T})$
  and every  $k < D/5$, we have
    \[
    \cond(S) \geq  2s(k) - 3^{k}\mu(S)^{\frac{1}{2}}.
    \]
\end{corollary}

In particular, as $s(k)$ tends to $1/2$, the expansion of small sets tends
to $1$. This corollary together with \pref{cor:manyeigs} completes the
proof of \pref{thm:main}.


\subsection{A Canonical Tester for Reed Muller codes}\label{sec:cantestrm}
We instantiate the construction from the previous section for the Reed Muller code. Let $\mc{C} =
\RM(n,n-d-1)$ be the Reed Muller code on $n$ variables of degree $n-d-1$,
which has $N = 2^n$, $H = \sum_{j \leq d} {n \choose j}$ and $D =2^{d+1}$.
Bhattacharyya, Kopparty, Schoenebeck, Sudan and
Zuckerman~\cite{BhattacharyyaKSSZ10} analyze the canonical tester
$\cT_\RM$ which samples a random minimum weight codeword from $\mc{C}^\perp$.
It is well known that the dual of $\RM(n,n-d-1)$ is exactly $\RM(n,d)$ and that
the minimum weight codewords in $\RM(n,d)$ are products of $d$ linearly
independent affine forms. They have weight $2^{n-d} = \eps N$ where $\eps =
2^{-d}$.  Thus, our graph $\Cay_\RM = \Cay(\RM_{n,d},\cT_\RM)$ has  as its vertices the $d$-degree polynomials
over $\GF 2^n$ with an edge between every pair of polynomials $P,Q$
such that $P-Q$ is equal to a minimum weight codeword, which are known to be products of
$d$ linearly independent affine forms.

\begin{theorem}[\cite{BhattacharyyaKSSZ10}]
  \label{thm:BKSSZ}
\label{thm:rm-tester-soundness}
  There exists a constant $\eta_0 > 0$ such that for all $n,d$, and
  $k<\eta_0 2^{d}$ the tester $\cT_\RM$ described above has soundness $s(k)
  \geq (k/2) \cdot 2^{-d}$.
\end{theorem}
\Dnote{say that \pref{lem:smooth} implies a better lower bound for the
  soundness of the \cite{BhattacharyyaKSSZ10} test, arbitrarily close to $k\cdot 2^{-d}$}

\pref{thm:BKSSZ} allows us to estimate the eigenvalue profile of
$\Cay_\RM$ and shows that small sets have expansion close to
$O(\eta_0)$. From here, we can get near perfect expansion by taking
short random walks. To avoid cumbersome
discretization issues we work with continuous time random walks on
graphs instead of the usual discrete random walks.

\begin{definition}
  For a graph $G$ the continuous-time random walk on $G$ with parameter $t$ is described by the (stochastic) matrix $G(t) = e^{-t(I-G)}$. $G(t)$ and $G$ have the same eigenvectors and the eigenvalues of $G(t)$ are $\set{e^{-t(1-\mu_i)}}$, where $\set{\mu_i}$ is the spectrum of $G$.
\end{definition}

We will view $\Cay_\RM(t)$ as a weighted graph.
We show that its eigenvalue profile is close to that of the noisy
cube, this stronger statement will be useful later.
\PGnote{This proof is a rather tedious calculation, perhaps we should
  send it to the appendix?}
%
%
\Dnote{do we need both cases of the lemma? the current statement looks rather scary}
\Mnote{I agree with PG. We need both cases of the lemma for Majority is stablest, I think.}
\begin{lemma}
  \label{lem:rm-eigenvalues}
Let $t = \e 2^{d+1}$ for $\e > 0$ and $\rho = e^{-\epsilon}$. Let
$\set{\lambda_\alpha}$ denote the eigenvalues of $\Cay_\RM(t)$.
\begin{itemize}
\item If $\deg(\chi_\alpha) = k$,  $\lambda_\alpha \leq
  \max(\rho^{k/2},\rho^{\mu_02^d})$ where $\mu_0$ is an absolute constant.
\item For all $\delta < \delta_0$ for some constant $\delta_0$, if $\deg(\chi_\alpha)=k < \delta^22^{d+1}$, $\abs{\lambda_\alpha - \rho^k}   \le \delta$.

\end{itemize}
\end{lemma}
\begin{proof}

Let $\set{\mu_\alpha}$ be the eigenvalues of $\Cay_\RM$ corresponding
to the character $\chi_\alpha$ so that $\lambda_\alpha = e^{-t(1-\mu_\alpha)}$.
Let $\tau=2^{-(d+1)}$. Since the canonical tester $\cT_\RM$ for $\cC$ is $2$-smooth, by \pref{lem:smooth},
  \begin{math}
    \mu_\alpha = 1-k\tau\pm k^2\tau^2.
  \end{math}
  Hence, $\lambda_\alpha=e^{-t(1-\mu_\alpha)}= e^{-\e k(1 \pm k
    \tau)} = \rho^ke^{-\e k^2 \tau}$.

For $k \leq 2^d$, $1-\mu_\alpha =  k\tau \pm k^2
\tau^2 \ge k\tau/2$.
Therefore, if $\deg(\chi_\alpha) \le 2^d$,
$\lambda_\alpha = e^{-\epsilon 2^{d+1}(1-\mu_\alpha)} \le
 e^{-\e 2^{d+1}k\tau/2} =  \rho^{k/2}$. For $k > 2^d$, by \pref{cor:spectrumbound},
$\mu_\alpha < 1 - 2s(k) < C_0$ for a universal constant $C_0 <
1$. Therefore, $|\lambda_\alpha| < e^{-\epsilon 2^{d+1}(1-C_0) } =
\rho^{(1-C_0) 2^{d+1}} < \rho^{\mu_02^d}$ for $\mu_0 < (1-C_0)/2$.

  %
  %

We now prove the second bound.
  If $\e k^2 \tau < \delta/10$, we have $\lambda_\alpha
  =\rho^{-k}(1\pm \delta)$ which implies $\abs{\lambda_\alpha - \rho^k} \le \delta$.
  Otherwise, if $\e k^2\tau\ge \delta/10$, our assumption $k< \delta^22^{d+1}$
  implies $\e k > 1/(10\delta)$, hence $\e^{-\e k} \leq
  e^{-\frac{1}{10 \delta}} \leq \delta/4$ for all
  $\delta < \delta_0$. For $k \leq 2^d$, $(1-\mu_\alpha) =  k\tau \pm k^2 \tau^2 \ge k\tau/2$. Hence
  $\lambda_\alpha \leq e^{-tk\tau/2} \leq    e^{-\frac{1}{20 \delta}}
  \leq \delta/4$ for all $\delta < \delta_0$.
  In this case, we get $\abs{\lambda_\alpha-\rho^k} \le |\lambda_\alpha| + |\rho^k| \le \delta/2$.

\end{proof}

Since the eigenvectors stay the same, $\Cay_\RM(t)$ inherits
the hypercontractive properties of $\Cay_\RM$. In particular, by Lemma \ref{lem:hc42},
$\norm{\cV}_{2\to 4} \leq 3^{\ell/2}$ where $\cV$ denotes polynomials
of degree $\ell \leq \frac{D-1}{4}$.
Combining Lemmas \ref{lem:hyper-to-sse} and \ref{lem:rm-eigenvalues},
we obtain a graph with small set expansion and many large
eigenvalues.

\eat{
\begin{theorem}
\label{thm:RM-graph-power}
For any $\eps > 0$, $\ell \in \N$, there exists a graph $G$ with $2^{(\log |G|)^{O(1/\ell)}}$ eigenvalues larger than $1-\e$ and every set $S \subseteq G$ has expansion
\[
\cond(S) \geq 1 - e^{-\e \ell/2} - 3^{\ell/2}\sqrt{\mu(S)}
\]
\end{theorem}
}

\begin{theorem}
\label{thm:RM-graph-power}
For any $\eps, \eta > 0$, there exists a graph $G$ with $2^{(\log
  |G|)^\frac{1}{d}}$ eigenvalues larger than $1-\e$ for $d =
  \log(1/\eps) + \log\log(1/\eta) +O(1)$ and where every set $S \subseteq G$ has expansion
\[
\cond(S) \geq 1 - \eta - 3^{\frac{c_1}{\eps}\log(1/\eta)}\sqrt{\mu(S)}
\]
for some constant $c_1$.
\end{theorem}
\begin{proof}
Let $\mu_0, \delta_0$ be constants from previous lemma.
Fix $\ell = \frac{c_1}{\e}\log(\frac{1}{\eta})$ so that $e^{-\e
  \ell/2} = \eta$ and $d= \log(\ell) + c_2$ so that $\ell \leq
\min(\mu_0 2^{d+1}, 2^d/5)$.  Consider the graph $\Cay_\RM(t)$ of the
continuous random walk on $\Cay_\RM$ where $t = \epsilon 2^{d+1}$ as
in \pref{lem:rm-eigenvalues}.  Note that the graph has $|G| = \sum_{j
  \leq d} \binom{N}{i}$ vetices. Let $\{\mu_\alpha\}$ be the spectrum of $\Cay_\RM$ and $\lambda_\alpha$ be the spectrum of $\Cay_\RM(t)$.

Then, for every $\alpha \in \GF2^N/\cC$, $\deg(\chi_\alpha) = 1$, we
have $s_\tau(\alpha) \geq 2^{-d}$. Hence
$\mu_\alpha \geq 1 -2^{-d+1}$, $\lambda_\alpha  \geq e^{-t2^{-d+1}} = e^{-4\e} $. Therefore, there are
at least $N = 2^{(\log |G|)^{1/d}}$ eigenvalues which are larger than
$1- 4 \epsilon$.

Since $\ell < \mu_02^{d+1}$, by \pref{lem:rm-eigenvalues} if $\deg(\chi_\alpha) > \ell$,
$\lambda_\alpha \leq \eta$.
Let $\cV$ be the subspace spanned by characters of degree at most
$\ell$. Since $\ell < 2^d/5$ by \pref{lem:hc42}, $\|\cV\|_{2\to 4} \leq 3^{\ell/2}$. Therefore, by \pref{lem:hyper-to-sse}, for any set $S \subseteq G$ with $\mu(S) \leq \delta$,
\begin{align*}
  \cond(S) \geq 1 - \eta - 3^{\frac{c_1}{\e}\log(1/\eta)} \sqrt{\mu(S)}.
\end{align*}
\end{proof}

\eat{
Plugging this into Theorem \ref{thm:main} gives the following
result. Since the tester is smooth (it queries each coordinate with
probability $\eps$), we get a slightly sharper result.
\begin{theorem}
\label{thm:RM-graph}
There is an absolute constant $\eta > 0$ such that for any $\eps > 0$
and large enough $n$, there exists a graph  on
$2^{n^{\log(1/\eps)}}$ vertices which has $2^n$ eigenvalues larger
than $1 - 2\eps$. Further, sets of fractional size $\mu <
\frac{1}{9^{\frac{\eta}{\eps}}}$ expand by $\eta$.
\end{theorem}

By amplifying the rejection probability of the tester (which corresponds to
taking powers of the graph), we get the following result (here we raise the
\RM\ graph to the power $t=\e 2^d$):
\begin{theorem}
\label{thm:RM-graph-power}
There is an absolute constant $\eta > 0$ such that for any $\eps > 0$, $t\in\N$
and large enough $n$, there exists a graph  on
$2^{n^{\log(2t/\eps)}}$ vertices which has $2^n$ eigenvalues larger
than $1 - \eps$ and every set $S$ has expansion
\[
\cond(S) \geq 1 - \eta^t - 3^{2^d}\sqrt{\mu(S)}
\]
\end{theorem}
}


\begin{remark}[Coding application]
The fact that our graph is a Cayley graph over $\GF 2^n$ has a potentially  interesting implication for coding theory. By looking at the set of edge labels as the rows of a generating matrix for a code, we know that large Fourier coefficients corresponds to low weight codewords, and hence we get a code of dimension $m=\binom{n}{d}$ that has an almost exponential (i.e. $2^n$) number of codewords of low weight, but yet has small \emph{generalized Hamming distance} in the sense that every subspace of codimension $\omega(1)$ contains a codeword of fractional Hamming weight $1 - o(1)$. In particular by setting $d$ to be a function slowly tending to infinity we can get a linear code for which correcting from an $o(1)$ fraction  of \emph{corruption} errors  requires an almost exponential list size, but for which one can correct a fraction approaching $1$ of \emph{erasure} errors using a list of constant size. (The code obtained by  taking all edges of our graph has an almost exponential blowup, but this can be reduced by subsampling the edges.) \Bnote{to a linear number? polynomial? should check}
\Mnote{This is a bit confusing. We should either elaborate on this, or not have it here.}
\end{remark}

\section{Majority is Stablest over Codes} \label{sec:stablestcodes}
In this section we show an analogue of the ``Majority is Stablest'' result of Mossel et al.~for the $\RM$ graph we constructed in the previous section; this will help us replace the noisy cube with the $\RM$ graph in various unique games gadgets.

We first review some definitions. For a function $f:\sbits^N \to \R$ and $\ell > 0$, define
\[ \Inf_i ^{\leq \ell}(f) = \sum_{\alpha \in \bits^N, |\alpha| \leq \ell, \alpha_i = 1}|\hat{f}(\alpha)|^2.\]
For $\rho > 0$, let $\Gamma_\rho:[0,1] \to [0,1]$ be the Gaussian noise stability curve defined as follows. For $\mu \in [0,1]$, let $t \in \R$ be such that $\Pr_{g \leftarrow \cN(0,1}[g < t] = \mu$. Then, $\Gamma_\rho(\mu) = \Pr_{X,Y}[ X \leq t, Y \leq t]$, where $(X,Y) \in \R^2$ is a two-dimensional mean zero Gaussian random vector with covariance matrix $\left(\begin{array}{cc} 1 &\rho\\ \rho &1\end{array}\right)$. We refer the reader to Appendix B in Mossel et al.~\cite{MosselOO05} for a more detailed discussion on $\Gamma_\rho$.

Let $P(x) = \sum_{I \subseteq [N]} a_I \prod_{i \in I} x_i$ be a $N$-variate multilinear polynomial $P:\R^N \to \R$. Define $\|P\|^2 = \sum_I a_I^2$ and we say $P$ is $\epsilon$-regular if for every $i \in [N]$, $\sum_{i \ni I} a_I^2 \leq \epsilon^2 \cdot \|P\|^2$.

Throught this section, we let $T_{\rho,N}$ (we omit $N$ when the dimension is clear) denote the noisy hypercube graph with second largest eigenvalue $\rho$.

\subsection{Majority is stablest and Invariance}
\label{sec:major-stabl-hyperc}

The following theorem shows that, in the context of noise stability, a
\emph{regular} function on the hypercube behaves like a function on Gaussian space.
\begin{theorem}[Majority is stablest, \cite{MosselOO05}]
  \label{thm:majority-stablest}
  Let $f\from \sbits^N\to [0,1]$ be a function with $\E f = \mu$.
  Suppose $\Inf_i^{\le 10\log(1/\tau)}(f)\le \tau$ for all $i\in[N]$.
  Then,
  \begin{displaymath}
    \iprod{f,T_\rho f}
    \le \Gamma_\rho(\mu)
    + \tfrac{10\log\log(1/\tau)}{(1-\rho) \log (1/\tau)}
    \mcom
  \end{displaymath}
  where $T_\rho$ is the Boolean noise graph with second largest eigenvalue
  $\rho$ and $\Gamma_\rho$ is the Gaussian noise stability curve.
\end{theorem}

We will need the following ingredient of the proof of
\pref{thm:majority-stablest} from \cite{MosselOO05}.
For $a,b\in\R$, let $\zeta_{[a,b]}\from \R\to \Rnn$ be the functional
$\zeta_{[a,b]}(x) = \max\set{a-x,x-b,0}^2$.
For a real-valued random variable $X$, the expectation $\E \zeta(X)$ is the
$L_2^2$-distance of $X$ to the set of $[a,b]$-valued random variables (over
the same probability space as $X$).
We will be interested in the case $a=0$ and $b=1$.
For this case, we abbreviate $\zeta=\zeta_{[0,1]}$.

\begin{theorem}[Invariance Principle, {\cite[Theorem 3.19]{MosselOO05}}]
\label{thm:moo-invariance}
  Let $P$ be an $\tau$-regular $N$-variate real multilinear polynomial with degree at most
  $\ell$ and $\snorm{P}\le 1$.
  %
  %
  Then,
  \begin{displaymath}
    \Abs{\E_{x\in \sbits^N} \zeta\circ P(x) - \E_{y\sim \cN(0,1)^N|}
      \zeta\circ P(y)} \le 2^{O(\ell)}\sqrt\tau
    \mper
  \end{displaymath}
\end{theorem}

We will need the following corollary of 
that can handle functions that are not $[0,1]$-valued as in the theorem but just close to
$[0,1]$-valued functions.
\Dnote{the corollary should still be true with $\E \zeta\circ f\le
  1/\log(1/\tau)$.
Note sure if it would affect our final bounds.
}
\begin{corollary}
  \label{cor:approx-majority-stablest}
  Let $f\from \sbits^N\to \R$ be a function with $\E f = \mu$ and $\E \zeta
  \circ f \le \tau$.
  Suppose $\Inf_i f^{\le 30\log(1/\tau)}\le \tau$ for all $i\in[N]$.
  Then,
  \begin{displaymath}
    \iprod{f,T_\rho f}
    \le \Gamma_\rho(\mu)
    + \tfrac{40\log\log(1/\tau)}{(1-\rho) \log (1/\tau)}
    \mcom
  \end{displaymath}
  where $T_\rho$ is the Boolean noise graph with second largest eigenvalue
  $\rho$ and $\Gamma_\rho$ is the Gaussian noise stability curve.
  (Here, we assume that $\tau$ is small enough.)
\end{corollary}

\begin{proof}
  Let $f'$ be the closest $[0,1]$-valued function to $f$.
  Since $\norm{f-f'}\le \sqrt{\tau}$, it follows that $\Inf^{\le
    20\log(1/\tau)} _i f' \le \tau+O(\sqrt\tau) \ll \tau^{1/3}$
  and $\E f' \le \E f + \sqrt{\tau}$.
  Since $\iprod{f,T_\rho f}\le \iprod{f',T_\rho f'} +O(\sqrt\tau)$, the
  corollary follows by applying \pref{thm:majority-stablest} to the
  function $f'$.
  (Here, we also use that  fact that $\Gamma_\rho(\mu+\sqrt\tau)\le
  \Gamma_\rho(\mu)+2 \sqrt\tau$. See Lemma B.3 in \cite{MosselOO05}.)
  %
\end{proof}

We remark that although we specialize to Reed--Muller codes in this section, most of the arguments generalize appropriately to arbitrary codes with good canonical testers modulo a conjecture about bounded independence distributions fooling low-degree polynomial threshold functions. We briefly discuss  this in \pref{sec:invar-princ-low}.

\newcommand{\rmd}{\cC^\bot}
\newcommand{\rgta}{\rightarrow}
To state our version of ``Majority is Stablest'' we first extend the notion of influences to functions over Reed--Muller codes. For $n,d\in \N$, $N = 2^n$, let $\cC\sse \GF2^N$ be the Reed--Muller code $\RM(n,n-d-1)$ and let $\cC^\bot \sse \GF2^N$ be its dual $\RM(n,d)$. For the rest of this section we assume that a set of representatives corresponding to the minimum weight codeword in each coset is chosen for the coset space $\GF2^N/\cC$.
\begin{definition}\label{def:rminfluence}
  For a function $f:\rmd \rgta \R$ and $i \in [N]$, $\ell > 0$, the $\ell$-degree influence of coordinate $i$ in $f$ is defined by
\[ \Inf_i^{\le \ell}(f) = \sum_{\substack{%
      \alpha \in \GF2^N/\cC,\\ \card{\alpha}\le \ell, \alpha_i=1}}
  \hat f(\alpha)^2
  \mper \]
\end{definition}
\Dnote{I don't understand the above expression because $\alpha_i$ for $\alpha\in \GF2^N/\cC$ is not defined unless on picks representatives. unless $\ell$ is less than half the minimum distance there is no canonical choice of representatives. hence it doesn't makes sense to define the influence for larger $\ell$}
(Recall that the Fourier coefficient $\hat{f}(\alpha) = \E_{x \in \rmd}[\chi_\alpha(x)]$.) As all $\alpha$'s with weight less than half the distance of $\cC$ fall into different cosets of $\cC$, for $\ell < D/2$, the above expression simplifies to
\[ \Inf_i^{\le \ell}(f) = \sum_{\substack{%
      \alpha \in \GF2^N,\;\card{\alpha}\le \ell,\;\alpha_i=1}} \hat f(\alpha)^2 \mper \]
\eat{(Note that $\Inf^{\le \ell}_i f$ is the squared norm of the projection of $f$ into the span of the characters $\chi_\alpha$ with $\deg(\chi_\alpha)\le \ell$ such that the (unique) minimum weight representative of $\alpha+\cC$ is non-zero in coordinate $i$. The minimum weight representative is unique because $\ell$ is less than half of the distance of the code.)}
The sum of $\ell$-degree influences of a function $f$ can be bounded
as below.
\begin{lemma} \label{lem:sum-of-influences}
	For a function $f:\rmd \to \R$ and $\ell <D/2$
	$$ \sum_{i \in [N]} \Inf^{\le \ell}_{i}(f) \leq \ell \Var[f]$$
	where $\Var[f] = \E[f^2] -(\E[f])^2$ denotes the variance of
	$f$.
\end{lemma}
\begin{proof}
	The lemma is an easy consequence from the definition of
	$\Inf^{\le \ell}(f)$ and the fact that $\Var[f] = \sum_{\alpha
	\neq 0} \hat f(\alpha)^2$.  We include the proof for the sake
	of completeness.
	\begin{align*}
	\sum_{i \in [N]} \Inf_i^{\le \ell}(f) & = \sum_{i \in [N]} \sum_{\substack{%
      \alpha \in \GF2^N,\;\card{\alpha}\le \ell,\;\alpha_i=1}} \hat
      f(\alpha)^2 \\
      & = \sum_{\substack{%
      \alpha \in \GF2^N,\;\card{\alpha}\le \ell,\;\alpha \neq 0}}
      \card{\alpha} \hat f(\alpha)^2 \\
      & \leq \ell \sum_{\substack{%
      \alpha \in \GF2^N,\;\card{\alpha}\le \ell,\;\alpha \neq 0}}
		 \hat f(\alpha)^2  \leq \ell \Var[f]
         \end{align*}
\end{proof}
We are now ready to state the main result of this section generalizing the Majority is Stablest result to Reed-Muller codes. Let $\cT_\RM$ be the canonical tester for $\cC$ as defined in \pref{sec:cantestrm}.
\begin{theorem} \label{thm:rm-majority-stablest}
There exist universal constants $c,C$ such that the following holds. Let $G$ be a continuos-time random walk on the $\RM$ graph $\Cay(\rmd,\cT_\RM)$ with parameter $t = \epsilon 2^{d+1}$. Let $f \from \rmd \to [0,1]$ be a function on $\rmd$ with $\E_{x \sim \rmd}[f(x)] = \mu$ and $\max_{i \in [N]} \Inf_i^{\le 30 \log(1/\tau)}(f) < \tau$. Then, for $d > C \log(1/\tau)$,
  \begin{equation}
    \label{eq:rmmjs}
    \iprod{f,Gf} \le \Gamma_\rho(\mu) + \frac{c \log \log (1/\tau)}{(1-\rho)\log(1/\tau)},
  \end{equation}
where $\rho = e^{-\epsilon}$ and $\Gamma_\rho\from \R \to \R$ is the noise stability curve of Gaussian space.
\end{theorem}

The proof of the theorem proceeds in three steps. We first show that the eigenvalue profile of the graph $G$ is close to the eigenvalue profile of the Boolean noise graph (see \pref{lem:rm-eigenvalues}). We then show an invariance principle for low-degree polynomials (and as a corollary for {\em smoothed functions}) showing that they have similar behaviour under the uniform distribution over the hypercube and the uniform distribution over the appropriate Reed--Muller code. Finally, we use the invariance principle to translate the majority is stablest result in the hypercube setting to the Reed--Muller code. The above approach is similar to that of Mossel et al., who translate a majority is stablest result in the Gaussian space to the hypercube using a similar invariance principle.

We first state the invariance principle that we use below (see the next subsection for the proof). Recall the definition of the functional $\zeta\from \R \to \R$ from \pref{sec:major-stabl-hyperc}.
\begin{theorem}\label{thm:iprmbounded}
Let $N = 2^n$ and $d \geq 4\log(1/\tau)$. Let $P:\R^N \rightarrow \R$ be a $\tau$-regular polynomial of degree at most $\ell$. Then, for $x \in_u \sbits^N$, $z \in_u \RM(n,d)$,
\[\left|\E[\zeta \circ P(x))] - \E[\zeta \circ P(z))] \right| \leq 2^{c_1 \ell} \sqrt \tau, \]
for a universal constant $c_1 > 0$.
\end{theorem}
The (somewhat technical) proof of \pref{thm:rm-majority-stablest} from the above invariance principle closely follows the argument of Mossel et al.~and is defered to the appendix -- see \pref{sec:missingproofs}. \Mnote{Moved the proof of MOS from IP to appendix as it is technical and we are following the same step in MOO closely.}
\subsection{Invariance Principles over Reed--Muller Codes}
The various invariance principles of Mossel et al. \cite{MosselOO05} are essentially equivalent (upto some polynomial loss in error estimates) to saying that for any low-degree regular polynomial $P$, the polynomial threshold function (PTF) $\sign(P(\;))$ cannot distinguish between the uniform distribution over the hybercube and the standard multivariate Gaussian distribution $\cN(0,1)^N$.
\begin{theorem}\label{th:ipmoo}
  Let $P:\R^N \rightarrow \R$ be a $\eps$-regular polynomial of degree at most $\ell$. Then, for any $x \sim\sbits^N$, $y \leftarrow \cN(0,1)^N$,
\[\left|\E[\sign(P(x))] - \E[\sign(P(y))] \right| \leq O(\ell \eps^{1/(2\ell+1)}).\]
\end{theorem}

Ideally, we would like a similar invariance principle to hold even when $x$ is chosen uniformly from the codes of the earlier sections instead of being uniform over the hypercube. Such an invariance principle will allow us to analyze alphabet reductions and integrality gaps based on graphs considered in earlier sections (e.g., the $\RM$ graph). We obtain such generic invariance principles applicable to all codes modulo certain plausible conjectures on low-degree polynomials being fooled by bounded independence.

For the explicit example of Reed--Muller code we bypass the conjectures and directly show an invariance principle by proving that the uniform distribution over the Reed--Muller code fools low-degree PTFs. To do so, we will use the specific structure of the Reed--Muller code along with the pseudorandom generator (PRG) for PTFs of Meka and Zuckerman \cite{MekaZ10}. Specifically, we show that the uniform distribution over $\RM$ can be seen as an instantiation of the PRG of \cite{MekaZ10} and then use the latter's analysis as a blackbox. Call a smooth function $\psi:\R \rightarrow \R$ $B$-nice if $|\psi^{(4)}(t)| \leq B$ for every $t \in \R$.
\begin{theorem}\label{thm:iprm}
Let $N = 2^n$ and $d \geq \log \ell + 2\log(1/\epsilon) + 2$. Let $P:\R^N \rightarrow \R$ be a $\eps$-regular multi-linear polynomial of degree at most $\ell$. Let $x \leftarrow \cN(0,1)^N$, $z \sim\RM(n,d)$. Then, for every $1$-nice function $\psi:\R \rightarrow \R$,
\[ \left|\E[\psi(P(x))] - \E[\psi(P(z))] \right| \leq \ell^2 9^\ell \eps^2\mper\]
\end{theorem}
To prove the theorem, we first discuss the PRG construction of \cite{MekaZ10}. Let $t = 1/\e^2$ and $M = N/t$. Let $\cH:[N] \rightarrow [t]$ be a family of almost pairwise independent hash functions\footnote{A hash family $\cH$ is almost pairwise independent if for every $i \neq j \in [N]$, $a,b \in [t]$, $\Pr_{h \in_u \cH}[ h(i) = a \,\wedge h(j) = b] \leq (1+\alpha)/t^2$ for $\alpha = O(1)$.} and let $\cD \equiv \cD_{4\ell}$ be a $(4\ell)$-wise independent distribution over $\sbits^m$. The PRG of \cite{MekaZ10}, $G_{\cH,\cD}$, can now be defined by the following algorithm:
\begin{enumerate}
\item Choose a random $h \in \cH$ and partition $[N]$ into $t$ blocks $B_1,\ldots,B_t$, with $B_j = \{i: h(i) = j\}$.
\item Choose independent samples $x_1,\ldots,x_t \leftarrow \cD$ and let $y \in \sbits^N$ be chosen according to an arbitrary distribution independent of $x_1,\ldots,x_t$.
\item Output\footnote{The description we give here is slightly different from that of \cite{MekaZ10} due to the presence of the string $y$. However, the analysis of \cite{MekaZ10} works without any changes for this case as well.}
  \begin{equation}
    \label{eq:mzgendef}
z' \in \sbits^N, \text{ with $z'_i  = z_i \cdot y_i$ for $i \in [N]$ }, \text{ where $z_{|B_j} = x_j$ for $j \in [t]$.}
  \end{equation}

\end{enumerate}
Meka and Zuckerman show that $G_{\cH,\cD}$ as above fool (arbitrary) low-degree polynomials. Below we state their result for regular PTFs which suffices for our purposes and gives better quantitative bounds.
\begin{theorem}\label{thm:mzgenip}[Lemma 5.10 in \cite{MekaZ10b}]
  Let $P:\R^N \rightarrow \R$ be a $\eps$-regular multilinear polynomial of degree at most $\ell$. Then, for $x \in_u \sbits^N$, and $y \in \sbits^N$ generated according to $G_{\cH,\cD}$,
\[ \left|\E[\psi(P(x))] - \E[\psi(P(y))] \right| \leq \frac{1}{3} \ell^2 9^\ell \eps^2\mper\]
\end{theorem}
We next show that the uniform distribution over $\RM(n,d)$ for a sufficiently high $d$ is equivalent to $G_{\cH,\cD}$ as above, for an appropriately chosen hash family $\cH$ and $(4\ell)$-wise independent distribution $\cD$. Below we identify $[N]$ with $\GF 2^n$ and $[t]$ with $\GF 2^c$, for $c = 2\log(1/\epsilon)$.
\begin{proof}[Proof of \pref{thm:iprm}]
For simplicity, in the following discussion we view $\RM(n,d)$ as generating a vector in $\GF 2 ^N$ and show that the uniform distribution over $\RM(n,d)$ has the appropriate independence structure as required by \pref{thm:mzgenip}, albeit with $\sbits$ replaced with $\bits$. This does not the effect the analysis of the generator.

Let $c = 2\log(1/\epsilon)$ and let $\cS$ be the subspace of polynomials of the form
\[ Q_1(x_1,\ldots,x_n) = \sum_{a \in \bits^c} \mathsf{1}(x_{|[c]} = a) \cdot P_a(x_{c+1},\ldots,x_n),\]
where the polynomials $P_a$ each have degree at most $d-c$. Note that we can sample a uniformly random element $Q_1 \in \cS$ by choosing independent, uniformly random degree at most $d-c$ polynomials $P_a:\GF 2^{n-c} \rightarrow \GF 2$ for $a \in \bits^c$ and setting $Q_1$ as above. This is because, each collection $(P_a)_{a \in \bits^c}$ leads to a unique element of $\cS$ and together they cover all elements of $\cS$.

Let $\cS'$ be a subspace of degree $d$, $n$-variate polynomials such that $\cS \cap \cS' = \{0\}$ and $\cS, \cS'$ together span all degree $d$ polynomials. Let $\cA:\GF 2^n \rightarrow \GF 2^n$ be the space of all affine transformations. For $A \in \cA$, let $h_A:[N] \rightarrow [t]$ be defined by $h_A(x) = A(x)_{|[c]}$ and let $\cH \equiv \{h_A: A \in \cA\}$. It is easy to see that for $A \in_u \cA$, the hash functions $h_A$ are almost pairwise independent. Observe that for $Q_1 \in_u \cS, Q_2 \in_u \cS'$ and $A \in_u \cA$, the polynomial $Q(\;) = (Q_1 + Q_2)(A(\;))$ is uniformly distributed over all $n$-variate degree $d$ polynomials.

Now, fix a polynomial $Q_2 \in \cS'$. Then, for a random $Q_1 \in_u \cS$, we have
\[ Q(x) = \sum_{a \in \bits^c} \mathsf{1}(h_A(x) = a) \cdot P_a(u_{a+1},\ldots,u_n) + Q_2(u),\]
where $u = Ax$ and the polynomials $(P_a)_{a \in \bits^c}$, are independent uniformly random polynomials of degree at most $d-c$ in $n-c$ variables. Let $\cD$ denote the distribution of $(P'(u))_{u \in \GF 2^{n-c}}$ for $P'$ a uniformly random polynomial of degree at most $d-c$ in $(n-c)$ variables. Then, for every fixed $A \in \cA$ and  $Q_2 \in \cS'$, the distribution of the evaluations of $Q$ restricted to different {\em buckets} $B_a = \{x:h_A(x) = a\}$ are independent of one another. Moreover, within each bucket $B_a$, the evaluations vector $(Q_1(x))_{x \in B_a}$ is distributed as $\cD$, which is $(2^{d-c}-1)$-wise independent.

Therefore, for every fixed $Q_2 \in \cS'$, the distribution of $z = (Q(x))_{x \in \GF 2^n}$ is the same as the output of $G_{\cH,\cD}$ as defined in Equation \ref{eq:mzgendef}, where $y = Q_2(A(x))$. The theorem now follows from Theorem \ref{thm:mzgenip}.
\end{proof}

The invariance principle of \pref{thm:iprm} combined with the appropriate choice of the smooth function $\psi$ gives us the following corollaries.
\begin{proof}[Proof of \pref{thm:iprmbounded}]
  Follows from using \pref{thm:iprm} and an argument as in Theorem 3.19 of \cite{MosselOO05} who get a similar conclusion for the hypercube starting from an invariance principle for the hypercube to the Gaussian space.
\end{proof}

\begin{corollary}\label{cor:iprmptf}
Let $N = 2^n$ and $d \geq \log \ell + 2\log(1/\epsilon) + 2$. Let $P:\R^N \rightarrow \R$ be a $\eps$-regular polynomial of degree at most $\ell$. Then, for $x \in_u \sbits^N$, $z \in_u \RM(n,d)$,
\[\left|\E[\sign(P(x))] - \E[\sign(P(z))] \right| \leq O(\ell \eps^{1/(2\ell+1)}). \]
\eat{
Similarly, for $y \leftarrow \cN(0,1)^N$,
\begin{equation}
  \label{eq:iprm}
  \left|\E[\sign(P(z))] - \E[\sign(P(y))] \right| \leq O(\ell \eps^{1/(2\ell+1)}).
\end{equation}}
\end{corollary}
\begin{proof}
Follows from \pref{thm:iprm} and Lemma 5.8 in \cite{MekaZ10}.
\end{proof}
\eat{
\begin{corollary}\label{cor:iprmbounded}
Let $N = 2^n$ and $d \geq \ell + 2\log(1/\epsilon) + 2$. Let $P:\R^N \rightarrow \R$ be a $\eps$-regular polynomial of degree at most $\ell$. Then, for $x \in_u \sbits^N$, $z \in_u \RM(n,d)$,
\[\left|\E[\zeta \circ P(x))] - \E[\zeta \circ P(z))] \right| \leq 2^{O(\ell)} \sqrt \eps. \]
\end{corollary}}
Finally a similar argument in the proof \pref{thm:iprm}, using a minor modification of the full analysis of the PRG from \cite{MekaZ10} (Theorem 5.17), shows that Reed--Muller codes with $d = \Omega(\ell \log(1/\epsilon))$ fool all degree $\ell$ PTFs. We exclude the proof in this work as we do not need the more general statement in our applications
\begin{theorem}\label{th:rmfoolptf}
There exists a constant $C > 0$ such that the following holds. Let $N = 2^n$ and $d = C \ell \log(1/\epsilon)$. Let $P:\R^N \rightarrow \R$ be a multilinear polynomial of degree at most $\ell$. Then, for $x \in_u \sbits^N$, $z \in_u \RM(n,d)$,
\[\left|\E[\sign(P(x))] - \E[\sign(P(z))] \right| \leq \epsilon.\]
\end{theorem}

\subsection{Invariance Principles over Codes}
\newcommand{\du}{\diffmacro u}
\label{sec:invar-princ-low}
Our main tool for proving that ``majority is stablest'' result over Reed--Muller codes, \pref{thm:rm-majority-stablest}, was the invariance principle \pref{thm:iprmbounded}. We conjecture that similar results should hold for any linear code with sufficiently large dual distance so that the codewords have bounded independence. In particular, we conjecture that bounded independence fools arbitrary
low-degree polynomial threshold functions (PTFs) over $\sbits^n$. \Mnote{Who conjectured this first?}

The conjecture is known to be true for halfspaces \cite{DiakonikolasGJSV09}, degree two PTFs
\cite{DiakonikolasKN10} and for Gaussians with bounded independence \cite{Kane11c}.

\begin{conjecture}
  \label{conj:bounded-independence-fools-ptfs}
  For all $d\in \N$ and $\e>0$, there exists $k=k(d,\e)$ such that the
  following holds:
  Let $Q$ be an $n$-variate multilinear real polynomial with degree $d$.
  Let $X$ be an $k$-wise independent distribution over $\sbits^n$ and let
  $Y$ be the uniform distribution over $\sbits^n$.
  Then,
  \begin{math}
    \abs{\E \sign\circ Q(X) - \E \sign \circ Q(Y)} \le \e\mper
  \end{math}
\end{conjecture}

Finally, we remark that for the application to ``majority is stablest'' it suffices to show a weaker invariance principle applicable to the $\zeta$ functional. 

\begin{conjecture}
  \label{conj:bounded-independence-preserves-boundedness}
  For all $d\in \N$ and $\e>0$, there exists $k=k(d,\e)$ and $\eta=\eta(\e)$
  such that the following holds:
  Let $Q$ be an $n$-variate multilinear real polynomial with degree $d$.
  Let $X$ be a $k$-wise independent distribution over $\sbits^n$
  and let $Y$ be the uniform distribution over $\sbits^n$.
  Suppose that $\E Q(X)^2\le 1$ and $\E\zeta\circ Q(X ) \le \eta$.
  Then, $\E \zeta\circ Q(Y)\le \e$.
\end{conjecture}

We show in the appendix that \pref{conj:bounded-independence-fools-ptfs} implies \pref{conj:bounded-independence-preserves-boundedness}.

\begin{lemma}\label{lem:fooltobounded}
  Let $X$ be a $20\ell$-wise independent distribution over $\sbits^N$ that
  $\e$-fools every $\tau$-regular degree-$\ell$ PTF.
  Then, for every $\tau$-regular $N$-variate multilinear real polynomial
  $Q$ with degree at most $\ell$ and $\E Q(X)\le 1$, we have for the
  uniform distribution $Y$ over $\sbits^N$,
  \begin{displaymath}
    \E \zeta \circ Q(Y) \le \E \zeta\circ Q(X) +
    2^{O(\ell)}\e^{0.9}
    \mper
  \end{displaymath}

\end{lemma}

\newcommand{\testt}{\cT_{t,\e}}

\section{Efficient Alphabet Reduction}
\label{sec:twoquerytest}
\label{sec:alphabetreduction}

\newcommand{\rmdt}{\cD^{t}}

The \emph{long code} over a (non-binary) alphabet $Q$ consists of the set
of dictator functions
\begin{math}
  \set{f_1,\ldots,f_N \from Q^N\to Q},
\end{math}
where $f_i(x)=x_i$ for all $x\in Q^N$.

A natural $2$-query test for this code was proposed by Khot \etal
\cite{KhotKMO07} and analyzed in Mossel \etal \cite{MosselOO05}.
The queries of the test are associated with the edges of the
\emph{$\e$-noise graph} on $Q^N$.
In this graph, the weight of an edge $(x,y)$ is its probability in the
following sampling procedure: Sample $x \in Q^N$ uniformly at random and
resample each coordinate of $x \in Q^N$ independently with probability
$\epsilon$ to generate $y \in Q^N$.

In this section, we present a more efficient code that serves as an analogue
for the long code over a non-binary alphabet.
For $n,d \in \N$, let $N = 2^n$ and let $\CC \subseteq \GF2^N$ be the
Reed--Muller code $\RM(n,n-d-1)$ and let $\cD=\cC^\perp \in \GF2^N$ be
its dual $\RM(n,d)$.
Let $\cT\sse \cD$ denote the canonical test set for the code $\cC$ as in \pref{sec:cantestrm}.
%
%

Let $t\in \N$ and let $Q = \GF2^t$.
We define the following distribution $\cT_t$ over $\cD^t$ (the $t$-fold
direct sum of $\cD$, a subspace of $\GF2^{t\cdot N}$),
\begin{itemize}

\item Sample $c$ from the test set $\cT\sse \cD$.
\item Sample $w=(\super w1,\ldots,\super w t)$ from $\GF2^t$ at random.
\item Sample $z = (\super z 1,\ldots,\super z t) \in \cD^t$ by
  setting
  \begin{displaymath}
    z^{(i)} =
    \begin{cases}
      c & \text{if } \super w i=1\mcom \\
      0 & \text{if } \super w i=0\mper
    \end{cases}
  \end{displaymath}
\end{itemize}
Consider the continuous-time random walk on the graph $\Cay(\cD^t,\cT_t)$
with parameter $\e\cdot 2^{d}$ (starting in point $0\in\cD^t$).
Let $\cT_{\e,t}$ be the distribution over $\cD^t$ corresponding to this
random walk.
The Cayley graph $\Cay(\cD^t,\cT_{\e,t})$ will serve us as an analogue of
the $\epsilon$-noise graph on $Q^N$.

\paragraph{Spectrum}

In the following we will demonstrate that (part of) the spectrum of the
Cayley graph $\Cay(\cD^t,\testt)$ corresponds to the spectrum of the
$\epsilon$-noise graph on $Q^N$.
To this end, we recall the spectrum of the $\epsilon$-noise graph on $Q^N$.
First, we define a convenient basis for the functions on $Q=\GF2^t$.
We will denote the coordinates of a vector $\alpha \in Q = \GF2^t$
by $\alpha = (\alpha^{(1)},\ldots,\alpha^{(t)})$.
The set of characters of $\GF2^t$ is $\set{\chi_\alpha\from \GF2^t\to
  \sbits\mid \alpha\in \GF2^t}$, where
\begin{displaymath}
  \chi_{\alpha} (x) = (-1)^{\sum_{j} \super \alpha j \super x j}\mper
\end{displaymath}

Since the noise graph on $Q^N$ is a Cayley graph over the abelian group
$\GF2^{tN}$, the characters of this group form a basis of eigenfunctions.
For $\beta = (\beta_1,\ldots,\beta_N) \in Q^N$, let $\chi_\beta\from Q^N\to
\sbits$ denote the character
\begin{displaymath}
 \chi_{\beta}(x_1,\ldots,x_N) = \prod_{i \in [N]} \chi_{\beta_i}(x_i) \mper
\end{displaymath}
The eigenvalue of $\chi_\beta$ in the $\e$-noise graph on $Q^N$
$(1-\epsilon)^{\wt(\beta)}$ where $\wt(\beta) = \abs{\set{i \mid \beta_i \neq
0^t}}$ is the Hamming weight of $\beta$ as a length-$N$ string over
alphabet $Q$. (In this section $\wt(\cdot)$ will always refer to the Hamming weight of
strings over alphabet $Q$.)

The canonical eigenfunctions of $\Cay(\cD^t,\cT_t)$ and
$\Cay(\cD^t,\testt)$ are indexed by $\beta \in Q^N/\cC^t$.
(Note that $\cC^t$ is the orthogonal complement of $\cD^t$.)
\Dnote{abbreviation is not used (anymore):
  For the sake of convenience, we just write $Q^N/\cC$ instead of
  $Q^N/\cC^t$.}
Analogous to the definition in \pref{sec:Cayley-graph-codes}, we define the
degree of a character $\chi_\beta\from \cD^t\to \sbits$ for $\beta \in
Q^N/\cC^t$ as,
\begin{displaymath}
\deg(\chi_\beta) = \wt(\beta)=\min_{\beta' \in \beta} \wt(\beta') \mcom
\end{displaymath}
where $\wt(\beta') = \card{\set{i\in[N] \mid \beta'_i \neq 0^t}}$ is the
Hamming weight of $\beta'$ seen as a length-$N$ string over alphabet $Q$.
(Here, the minimum is over all $\beta'\in Q^N$ that lie in the same coset as $\beta$ in $Q^N/\cC^t$.)
\Dnote{the notations for degree and weight in this paragraph are somewhat
  ambiguous. strictly speaking one should index the notations by $Q$ (the
  alphabet we use for measuring the Hamming distance). }

The following lemma is an analogue of Lemmas \ref{lem:smooth} and \ref{lem:rm-eigenvalues} and shows that the eigenvalues of $\Cay(\cD^t,\cT_t)$ are similar to the 
eigenvalues of the $\epsilon$-noise graph.

\begin{lemma}
  Let $\beta \in Q^N/\cC^t$.
  The eigenvalue $\lambda_\beta$ of the character $\chi_{\beta}$ in the graph
  $\Cay(\cD^t,\cT_t)$ satisfies $\lambda_{\beta} = 1-\wt(\beta)/2^d\pm
  O(\wt(\beta)/2^d)^2$
  and $\lambda_\beta\le 1-\Omega(1/t)\cdot \min\set{\vbig\wt(\beta)\cdot 2^{-d},1}$.
\end{lemma}
\begin{proof}
  We will first prove an upper bound on $\lambda_\beta$ for the case that
  $\wt(\beta)\gg 2^d$.
  We write $\beta=(\super \beta1,\ldots,\super \beta t)$ with $\super \beta i\in
  \GF2^N$.
  Let $z=(\super z1,\ldots,\super zt)\in \cD^t$ be a string drawn from the
  distribution $\cT_t$.
  Note that $\super z i = \super w i\cdot c$, where $w=(\super
  w1,\ldots,\super w t)$ and $c$ are sampled as in the definition of
  $\cT_t$.
  Since $w$ is a random vector in $\GF2^t$, we can upper bound
  $\lambda_\beta$,
  \begin{align*}
    \lambda_\beta &= \E_{z}(-1)^{\iprod{\beta,z}} \\
    &=   1-2\Prob[w\in \GF2^t,~c\in\cT]{\tsum_{i=1}^t \super wi\iprod{\super \beta i,c}=1}\\
    &= 1-\Prob[c\in \cT]{\exists i.~\iprod{\super \beta i,c}=1}\\
    &\le 1-\max_{i\in[t]} \Prob[c\in \cT]{\iprod{\super \beta i,c}=1}
    \mper
  \end{align*}
  Without loss of generality, we may assume that $\super \beta t$ has
  Hamming weight (as a binary string) at least $\wt(\beta)/t$.
  By \pref{thm:BKSSZ}, if $\wt(\beta)>\eta 2^{-d}$ for sufficiently small
  $\eta>0$, we can upper bound $\lambda_\beta\le 1-\Omega(\eta/t)$.

  Next, we will estimate $\lambda_\beta$ (from below and above) for
  $\wt(\beta)\ll 2^{-d}$.
  Let $I\sse [N]$ be the set of coordinates $i\in [N]$ with $\beta_i\neq
  0^t$.
  We claim,
  \begin{displaymath}
    \lambda_\beta = 1-\Prob[c]{\card{I\cap \supp(c)}=1}
    \pm O(1)\cdot\Prob[c]{\card{I\cap \supp(c)}\ge 2}\mper
  \end{displaymath}
  We write $\beta=(\beta_1,\ldots,\beta_N)$ with $\beta_i\in \GF2^t$.
  Then, $\iprod{\beta,z} = \sum_{i\in [N]} c_i \iprod{w,\beta_i}$.
  We refine the event $\iprod{\beta,z}=1$ according to the cardinality of
  $I\cap \supp(c)$.
  If $I\cap \supp(c)=\eset$, then $\iprod{\beta,z}=0$.
  On the other hand, conditioned on $\card{I\cap \supp(c)}$, the event
  $\iprod{\beta,z}=1$ is equivalent to the event $\iprod{w,\beta_{i_0}}=1$
  with $\set{i_0}=I\cap \supp(c)$.
  Since $\beta_{i_0}\neq 0^t$, this event has (conditional) probability
  $\half$.
  Hence,
  \begin{displaymath}
    \Prob[z]{\vbig \iprod{\beta,z}=1}
    = \tfrac12 \Prob[c\in \cT]{\vbig \card{I\cap \supp(c)} = 1 }
    \pm \Prob[c \in\cT]{\vbig \card{I\cap \supp(c)} \ge 2 }\mcom
  \end{displaymath}
  which implies the claimed estimate for $\lambda_\beta$.

  It remains to estimate the distribution of $\card{I\cap \supp(c)}$.
  The argument is similar to the proof of \pref{lem:smooth}.
  For every coordinate $i\in [N]$, we have $\Prob[c\in\cT]{c_i=1}=2^{-d}$.
  Thus, $\Prob{\,\card{I\cap \supp(c)}=1} \le \card{I}\cdot
  2^{-d}=\wt(\beta)/2^d$.
  On the other hand, for any two distinct coordinates $i\neq j\in [N]$, we
  have $\Prob[c\in\cT]{c_i=c_j=1}=2^{-2d}$.
  Therefore,
  \begin{displaymath}
    \Prob{\,\card{I\cap \supp(c)}=1}\ge
    \sum_{i\in I}\Prob{c_i=1}-\sum_{i<j\in I}\Prob{c_i=c_j=1}
    \ge \wt(\beta)/2^{d} -  (\wt(\beta)/2^d)^2.
  \end{displaymath}
  Similarly, $\Prob{\,\card{I\cap \supp(c)} \ge 2}\le (\wt(\beta)/2^d)^2$.
  We conclude that
  \begin{displaymath}
    \lambda_\beta = 1 - \wt(\beta)/2^d \pm O(\wt(\beta)/2^d)^2
  \end{displaymath}
  (Note that the estimate is only meaningful when $\wt(\beta)\ll 2^d$.)
\end{proof}

If the character $\chi_\beta$ has eigenvalue $\lambda_\beta$ in the graph
$\Cay(\cD^t,\cT_t)$, then it has eigenvalue $e^{-\e(1-\lambda_\beta)/2^d}$
in $\Cay(\cD^t,\testt)$.
Similarly to \pref{lem:rm-eigenvalues}, the eigenvalue of a character
$\chi_\beta$ is close to $e^{-\e\wt(\beta)}$ in the graph
$\Cay(\cD^t,\testt)$.

\begin{lemma}
  \label{lem:eigenvalues-larger-alphabet}
  \begin{itemize}
  \item  If $\wt(\beta)\le \delta^2 2^d$ for sufficiently small $\delta$, then the
  character $\chi_\beta$ has eigenvalue $e^{-\e\cdot \wt(\beta)}\pm \delta$ in
  the graph $\Cay(\cD^t,\testt)$.
  \item For an absolute constant $c_0$ and all $\beta \in Q^N/\rmdt$,  $\lambda_{\beta} \leq \max(\rho^{\wt(\beta)/c_0t},\rho^{
	  2^d/c_0t})$.
  \end{itemize}
  \Dnote{check what's the right bound on $\wt(\beta)$. will probably depend
  on $t$}
\end{lemma}

\paragraph{Influences}

Let $\beta\in Q^N/\cC^t$.
Suppose $\wt(\beta)< \wt(\cC^t)/2$.
(Note that $\cC^t\sse Q^N$ has the same minimum distance as $\cC\sse
\GF2^N$. \Mnote{This is not clear; you could lose a factor of $t$ when changing alphabet.})
In this case, we will identify $\beta$ with the (unique) codeword of
minimum weight in the equivalence class $\beta\in Q^N/\cC^t$.

\begin{definition}
  For a function $f \from \cD^t \to \R$, a coordinate $i \in [N]$, and a
  degree bound $\ell < \dist(\cC^t)/2)$, we define the \emph{$\ell$-degree
    influence of coordinate $i$ on $f$} as
  \begin{displaymath}
    \Inf_i^{\leq \ell}(f) = \sum_{\beta \in Q^N/\cC^t, ~\beta_i \neq 0^t,~
      \wt(\beta) \leq \ell}
    \hat{f}(\beta)^2 \mper
  \end{displaymath}
  (Here, $\beta_i$ refers to the $i$-th coordinate of the unique
  minimum-weight representative of the equivalence class $\beta$.)
\end{definition}

\subsection{Majority is Stablest}\label{sec:mosq}
In this section, we show an analogue of the majority is stablest theorem of \cite{MosselOO05} on the $\epsilon$-noise graph on $Q^N$ just as \pref{thm:rm-majority-stablest} showed an analogue of the majority is stablest theorem over the Boolean noise graph.
\begin{theorem} \label{thm:rm-majority-stablest-Q}
	\Pnote{just copied the parameters from the other majority is
	stablest}
For every $\epsilon,\delta,t > 0$, there exists $L,d,\tau$ such that if
$G$ denotes the graph $\Cay(\rmdt,\testt)$ constructed using Reed-Muller codes of
degree $d$, then
 for every function $f
\from \rmdt \to [0,1]$ with $\max_{i \in [N]} \Inf_i^{\leq L}(f) <
\tau$,
  \begin{equation}
    \label{eq:rmmjs-Q}
    \iprod{f,Gf} \le \Gamma_\rho(\mu) + \delta,
  \end{equation}
where $\rho = e^{-\epsilon}$, $ \mu = \E_{x \sim
\rmdt}[f(x)]$ and $\Gamma_\rho\from \R \to \R$ is the noise stability
curve over Gaussian space.
\end{theorem}
Given the characterization of the spectrum of $\Cay(\rmdt,\testt)$
(\pref{lem:eigenvalues-larger-alphabet}), the proof of
\pref{thm:rm-majority-stablest-Q} is similar to that of
\pref{thm:rm-majority-stablest}.  For the sake of completeness, we
include a proof sketch in the appendix -- see \pref{sec:mosqapp}. \Mnote{Moved the proof to the appendix as it is very similar to the earlier proof, which is also moved to the appendix.}re.

\subsection{$2$-Query Test}

We will now describe a dictatorship test for functions on $\rmdt$,
analogous to the $2$-query dictatorship test on $\epsilon$-noise
graph.

We are interested in functions $f\from \cD^t \to Q$ where $Q =
\GF2^t$.  Note that $v \in \cD^t$ can also be thought of as
$v \in Q^N$.  For all $\beta \in \GF2^n$, the {\em $\beta^{\th}$ dictator
function} $\chi_{\beta}$ from $\rmdt \subseteq Q^N$ to $Q$ is given by,
$$ \chi_{\beta}(c) = c_\beta $$
Clearly, the dictator functions are linear functions over
$\rmdt$, i.e., $\chi_{\beta}(c+c') = \chi_{\beta}(c) +
\chi_{\beta}(c')$.  This linearity is used to perform the $2$-query test 
via {\it folding}.  Note that for each $\alpha \in Q$, the constant
function $\alpha(x) =
\alpha$ for all $x \in \GF2^n$, belongs to the code $\rmdt$.  We will {\it fold} the function
by enforcing that for all $\alpha \in Q$, $f(c+ \alpha) = f(c) +
\alpha$ for all $\alpha \in Q$.

The details of the $2$-query dictatorship test is described below.
\begin{mybox}
	$\mathsf{DICT}$

	Input: $f\from \rmdt \to Q$

	\paragraph{Folding}  The function is assumed
to satisfy $f(c+r) = f(c) + r$  for every $c \in \rmdt$ and $r \in Q$.  This is enforced
by {\it folding} the table of the function $f$.

	\begin{itemize}
	\item Sample a vertex $c \in \rmdt$.

	\item Sample a neighbour $c' \in \rmdt$ of the vertex $c$ in the Cayley graph
		$\Cay(\rmdt,\testt)$.
		
	\item Sample $r \in Q$ uniformly at random.

	\item Accept if $f(c+r) - r   = f(c')$
	\end{itemize}
\end{mybox}

Given a function $f \from \rmdt \to Q$, we can arithmetize the value
of the test in terms of $Q$ functions $\{f_{\alpha}\}_{\alpha \in Q}$
that are defined as
$$ f_\alpha(x) = \Ind[f(x) = \alpha]  \mper$$
Due to folding, we have $f_\alpha(x) = f_{\alpha+r}(x+r)$ for all $r
\in Q$.  For each $\alpha \in Q$, the expectation of $f_{\alpha}$ is given by,
$$ \E_{c \in \rmdt} f_\alpha(c) = \Pr_{c\in \rmdt,r \in Q}
\left[f(c+r) = \alpha \right] = \frac{1}{Q} \mper$$
where we used the fact that $f$ is folded.
The probability of acceptance of the $2$-query test can be written
in terms of the functions $f_{\alpha}$ as follows:
$$ \Pr[\text{Test accepts} f] = \sum_{\alpha \in Q} \E_{(c,c') \sim
\Cay(\rmdt,\testt)}\left[f_{\alpha + r}(c+r)f_{\alpha}(c')\right] = \sum_{\alpha \in Q} \E_{(c,c') \sim
\Cay(\rmdt,\testt)}\left[f_{\alpha}(c)f_{\alpha}(c')\right],$$
where $(c,c') \sim \Cay(\rmdt,\testt)$ denotes a uniformly random edge in the graph $\Cay(\rmdt,\testt)$.

\begin{theorem} \label{thm:twoquerytest}
	The $2$-query dictatorship test $\mathsf{DICT}$ described
	above satisfies the following completeness and soundness,
	\begin{itemize}
	\item (Completeness)  Every dictator function
		$\chi_{\beta}(x) = x_i$ is accepted by the
		test with probability at least $1-\epsilon$.
	\item (Soundness) 		
		For every $\delta > 0$, there exists $\tau, L$
		such that if $f$ satisfies
		$\max_{i \in [N]} \Inf^{\leq L}_{i}(f_{\alpha}) \leq \tau$ for all
		$\alpha \in Q$ then $f$ is accepted with probability at
		most
		$$Q\cdot \Gamma_{\rho}\left(\tfrac{1}{Q}\right) + \delta \mcom$$
		where $\rho = e^{-\eps}$.
	\end{itemize}
\end{theorem}

\paragraph{Completeness}

Recall that for a $c \in \rmd$ generated from
distribution $\cT_{\epsilon}$, for each $x \in \GF2^n$ (see \pref{lem:eigenvalues})
$$\Pr_{c \sim \cT_{\epsilon}}\left[c(x) = 0\right] \geq
1-O(\epsilon) \mper$$
It is easy to see that by construction, this property holds for the
distribution $\testt$ also, namely,
$$\Pr_{c \sim \testt}\left[c(x) = 0\right] \geq
1-O(\epsilon) \mper$$
Hence for a random edge $(c,c')$ in the Cayley graph $\Cay(\rmdt,
\testt)$ and an $\beta \in \GF2^n$, $c(\beta) = c'(\beta)$ with probability
$1-\epsilon$.  Therefore, for each $\beta \in \GF2^n$, the $\beta$th
dictator function satisfies the test with probability $1-
O(\epsilon)$.

\paragraph{Soundness}
The probability of acceptance of the $2$-query test is given by,
$$Pr[\text{Test accepts} f] =  \sum_{\alpha \in Q} \E_{(c,c') \sim
\Cay(\rmdt,\testt)}\left[f_{\alpha}(c)f_{\alpha}(c')\right]
$$

By applying \pref{thm:rm-majority-stablest}, there is an appropriate choice of
$L,\tau$ such that if $\max_{i \in [N]} \Inf^{\leq L}_{i}(f_{\alpha})
\leq \tau$ for all $\alpha$ then the probability of acceptance can be bounded by
$$ \Pr[\text{Test Accepts}] =  \sum_{\alpha \in Q}
\iprod{f_{\alpha},G f_{\alpha}} \leq Q\cdot
\Gamma_{\rho}\left(\tfrac{1}{Q}\right) + \delta \mcom$$
where $\rho = e^{-\eps}$ and $G = \Cay(\rmdt,\testt)$.  The conclusion follows.

%

\section{Efficient integrality gaps for unique games} \label{sec:effug}

In this section, we present constructions of SDP integrality gap
instances starting from a code $\CC$ along with a local tester.  To this
end, we make an additional assumption on the code $\CC$.  Specifically, let us suppose there exists a subcode $\mc{H}$ of $\cD = \CC^{\perp}$ with distance $\frac{1}{2}$.
Formally, we show the following result.

\begin{theorem} \label{thm:code-to-uggap}
Let $\CC$ be an $[N,K,D]_2$ linear code with a canonical tester $\cT$
as described in \pref{def:canonical-tester}.  Furthermore, let
$\mc{H}$ be a subcode of $\cD = \CC^{\perp}$ with distance $\frac{1}{2}$.
Then, there exists an instance of unique games, more specifically a
$\mc{H}\dashmaxtwolin$ instance, whose vertices are
$\cD$ ($|\cD| = 2^{N-K}$) and alphabet $\mc{H}$ such
that:
\begin{itemize}
	\item The optimum value of the natural SDP relaxation for
		unique games is at least $\left(1 -
		\frac{2t}{N}\right)^2$ where $t$ is the number of
		queries made by the canonical tester $\cT$.
	\item No labelling satisfies more than $$ \min_{k \in [0,D/5]}
		\left(1 - 2 s(k) + \frac{3^k
		}{|\mc{H}|^{\frac{1}{2}}} \right)$$ fraction of
		constraints.
\end{itemize}
\end{theorem}

Instantiating the above theorem with the Reed--Muller code and its
canonical tester we obtain the following explicit SDP integrality gap
instance.
\begin{corollary} \label{cor:code-to-uggap}
	For every integer $n$, $\delta > 0$  there exists a
	$\GF2^n\dashmaxtwolin$ instance $\Gamma$ on
$M = 2^{2^{\log^2 n}}$ vertices such
that the optimum value of the SDP relaxation on $\Gamma$ is $1 -
O(\frac{\log(1/\delta)}{n}) = 1- O\left(\frac{\log(1/\delta)}{2^{(\log\log M)^{1/2}}}\right)$ while every labelling of $\Gamma$
satisfies at most $O(\delta)$ fraction of edges.
\end{corollary}
\begin{proof}Fix the code $\CC$ to be the Reed--Muller code
	$\RM(n,n-\log n)$ of degee $d = \log n$ over $n$ variables.  The
	block length of the code is $N = 2^n$, while the rate is
	$K = 2^n - \sum_{i \leq d} \binom{n}{i} \leq 2^n - O(2^{\log^2
	n})$.  This code contains the Hadammard code $\mc{H}$ which is of
	relative distance $\frac{1}{2}$.

	Let $\cT_{\RM}$ denote the canonical Reed--Muller tester for
	$\RM(n,n-\log n)$, and let $\cT_{\RM}^{\oplus r}$ denote the
	XOR of $r$-independent tests.  Let us fix $r =
	100\log(1/\delta)$, thus yielding a canonical tester making
	$t = \log(1/\delta) \cdot 2^{n-d}$ queries.  By the work of
	\cite{BhattacharyyaKSSZ10}, this tester has a soundness of at
	least $s(k) = \frac{1}{2} - (1-k/2^{d+1})^r/2$.  With $k =
	2^{d}/10$, the above soundness is at least $s(k) \geq 1/2 -
	\delta/2$.  Using \pref{thm:code-to-uggap}, the optimum value
	of the resulting $\GF2^n\dashmaxtwolin$ instance is at most
	$\delta$.  On the other hand, the SDP value is at least 
$$(1-2t/N)^2 = 1-100 \log(1/\delta) 2^{n-d}/2^n = 1-O\left(\frac{\log 1/\delta}{2^d}\right) = 1-\frac{\log (1/\delta)}{n}\mper$$
\eat{
\Mnote{May be state the previous theorem in terms of $\Phi(S)$ for the graph $\Cay(C^\perp,\cT)$ and use Theorem \ref{thm:RM-graph-power} with $t = O(\log(1/\delta))$, $\epsilon = 1/n$.}
\Pnote{Unfortunately, one will have to also define expnasion for
fractional sets, it might still be neater}}
\end{proof}

Starting from $\CC$, we construct an SDP integrality gap instance $\Gamma(\CC,\mc{T})$ for unique games as described below.

\begin{mybox}
The vertices of $\Gamma_{\CC}$ are the codewords of $\cD$.  The alphabet of the unique games instance $\Gamma(\CC,\mc{T})$ are the codewords in $\mc{H}$.  The constraints of unique games instance $\Gamma(\CC,\mc{T})$ are given by the tests of the following verifier.

The input to the verifier is a labeling $\ell : \cD \to
\mc{H}$.  Let us denote by $R = |\mc{H}|$.  The verifier proceeds as follows:
\begin{itemize}
	\item Sample codewords $\mb{c} \in \cD$ and $\mb{h},\mb{h'} \in \mc{H}$ uniformly at random.
\item Sample a codeword $\mb{q} \in \cD$ from the tester $\mc{T}$.
\item Test if
	$$ \ell(\mb{c} + \mb{q} + \mb{h}) - \ell(\mb{c}+\mb{h'}) =
	\mb{h} - \mb{h'} $$
\end{itemize}
\end{mybox}

\paragraph{SDP Solution}
Here we construct SDP vectors that form a feasible solution to a
natural SDP relaxation of unique games \cite{KhotV05}.

\begin{mybox}
\begin{align}
	\text{Maximize} & \E_{\mb{c} \in \cD, \mb{h},\mb{h'} \in \mc{H}} \E_{\mb{q} \in \mc{T}}
				\left[ \frac{1}{R} \sum_{\mb{\ell} \in \mc{H}}
				\iprod{ \vec
				b_{\mb{c+h'},\mb{\ell+h'}}, \vec
				b_{\mb{c+q+h},\mb{\ell+h}}}  \right]\\
	\text{Subject to} & \iprod{\vec b_{\mb{c},\mb{h}}, \vec
	b_{\mb{c},\mb{h'}}} = 0 & \forall c \in \cD,  h \neq h' \in
	\mc{H} \\
	& \iprod{\vec b_{\mb{c},\mb{h}},\vec b_{\mb{c}',\mb{h}'}} \geq
	0 & \forall c,c' \in \cD, h,h' \in \mc{H}. \\
	& \sum_{\mb{\ell} \in \mc{H}} \iprod{\vec
	b_{\mb{c},\mb{\ell}},\vec b_{\mb{c},\mb{\ell}}} = R & \forall
	c \in \cD
\end{align}
\end{mybox}

For a vector $\mb{c} \in \GF 2^m$, we will use $\vc{c} \in \R^m$ to
denote the vector whose coordinates are given by $\vc{c}_i =
(-1)^{\mb{c}_i}$.  For a pair of vectors $\mb{c},\mb{c'}$, we have
$$ \iprod{(-1)^{\mb{c}},(-1)^{\mb{c'}}} = 1-2 \Delta(\mb{c},\mb{c}')
\mper$$

For each vertex $\mb{c} \in \cD$ associate vectors $\{\vec
b_{\mb{c},\mb{h}} = \vc{c+h} \otimes \vc{c+h} | \mb{h} \in \mc{H}\}$.
Notice that for a pair of vectors $\vec b_{\mb{c},\mb{h}}, \vec
b_{\mb{c'}, \mb{h'}}$ we have,
$$ \iprod{\vec b_{\mb{c},\mb{h}}, \vec
b_{\mb{c'}, \mb{h'}}} = \iprod{ (-1)^{\mb{c}+\mb{h}}, (-1)^{\mb{c'} +
\mb{h'}}}^2 = (1-2\Delta(\mb{c}+\mb{h},\mb{c'}+\mb{h'}))^2 \mper
$$
Since the distance of the code $\mc{H}$ is $\frac{1}{2}$, we have
\begin{equation}
	\iprod{\vec b_{\mb{c},\mb{h}}, \vec b_{\mb{c},\mb{h}'}} = (1 -
	2 \Delta(\mb{h}, \mb{h}'))^2 = \begin{cases} 1 & \text{ if }
							\mb{h} = \mb{h}' \\
						0 & \text{ if } \mb{h} \neq \mb{h}'
					\end{cases}
\end{equation}
In other words, for every vertex $\mb{c}$, the corresponding SDP
vectors are orthonormal.
The objective value of the SDP solution is given by,
\begin{eqnarray*}
\mathsf{OBJ}
& = & \E_{\mb{c} \in \cD, \mb{h},\mb{h'} \in \mc{H}} \E_{\mb{q} \in \mc{T}}
		\left[ \frac{1}{R} \sum_{\mb{\ell} \in \mc{H}}
			\iprod{ \vec b_{\mb{c+h'},\mb{\ell+h'}}, \vec	b_{\mb{c+q+h},\mb{\ell+h}}}  \right] \\
	& = & \E_{\mb{c} \in \cD, \mb{h} \in \mc{H}} \E_{\mb{q} \in\mc{T}}
		\left[ \frac{1}{R} \sum_{\mb{\ell} \in \mc{H}} (1- 2
		\Delta(\mb{c+h'}+\mb{\ell+h'}, \mb{c+q+h}+\mb{\ell+h}))^2  \right] \\
& =& \E_{\mb{c} \in \cD, \mb{h} \in \mc{H}} \E_{\mb{q} \in \mc{T}}
		\left[(1- 2 \Delta(\mb{0}, \mb{q}))^2  \right] \\
& \geq & \left(1-\frac{2t}{N}\right)^2
\end{eqnarray*}
where $t$ is the number of queries made by the canonical tester $\cT$ for $\CC$.

\paragraph{Soundness}
Let $\ell : \cD \to \mc{H}$ be an arbitrary labelling of the Unique
Games instance $\Gamma(\CC, \mc{T})$.  For each $\mb{p} \in \mc{H}$, define
a function $f_{\mb{p}} \from \cD \to [0,1]$ as follows,
$$ f_{\mb{p}} (\mb{c}) = \E_{\mb{h} \in \mc{H}} \left[\Ind[\ell(\mb{c+h}) =
\mb{p+h}]\right] \mper $$

The fraction of constraints satisfied by the labelling $\ell$ is given
by,
\begin{eqnarray}
\mathsf{OBJ}
& = & \E_{\mb{c} \in \cD, \mb{h},\mb{h'} \in \mc{H}} \E_{\mb{q} \in \mc{T}}
		\left[  \sum_{\mb{p} \in \mc{H}}
			\Ind [\ell(\mb{c+h'}) = \mb{p+h'}] \cdot
			\Ind[\ell(\mb{c+q+h})= \mb{p+h}]  \right]
			\nonumber \\
& = & \E_{\mb{c} \in \cD} \E_{\mb{q} \in \mc{T}}
		\left[  \sum_{\mb{p} \in \mc{H}}
		\E_{\mb{h'} \in \mc{H}} \Ind [\ell(\mb{c+h'}) = \mb{p+h'}] \cdot
		\E_{\mb{h} \in \mc{H}}	\Ind[\ell(\mb{c+q+h})=
		\mb{p+h}]  \right] \nonumber \\			
& = & \E_{\mb{c} \in \cD} \E_{\mb{q} \in \mc{T}}
		\left[  \sum_{\mb{p} \in \mc{H}}
		f_{\mb{p}}(\mb{c}) f_{\mb{p}}(\mb{c+q})  \right] \\			
		& = & \sum_{\mb{p} \in \mc{H}} \iprod{f_{\mb{p}},G
		f_{\mb{p}}} \label{eq:fraction-of-constraints}
\end{eqnarray}
where $G=\Cay(\cC^\bot,\mc{T})$ is the graph associated with
the code $\CC^\bot$ and tester $\cT$.

The expectation of the function $f_{\mb{p}}$ is given by,
\begin{eqnarray}
\E_{\mb{c} \in \cD} f_{\mb{p}}(\mb{c})
& = & \Pr_{\mb{c} \in \cD,\mb{h} \in \mc{H}} \left[\ell(\mb{c+h}) =
\mb{p+h}\right] \nonumber \\
& = & \Pr_{\mb{c} \in
\cD,\mb{h} \in \mc{H}} \left[\ell(\mb{c}) =
\mb{p+h}\right] \textrm{  because  }(\mb{c +h},\mb{h}) \sim
(\mb{c},\mb{h}) \nonumber \\
& =& \frac{1}{|\mc{H}|} = \frac{1}{R} \mper \nonumber
\end{eqnarray}
Since $f_{\mb{p}}$ is bounded in the range $[0,1]$ we have,
$$ \iprod{f_{\mb{p}},f_{\mb{p}}} = \E_{\mb{c} \in
\cD}[f_{\mb{p}}(\mb{c})^2] \leq \E_{\mb{c} \in
\cD}[f_{\mb{p}}(\mb{c})] = \frac{1}{R} \mper$$
Applying \pref{cor:hc-sse}, we
get that for each $\mb{p}$,
$$ \iprod{f_{\mb{p}},G f_{\mb{p}}} \leq \frac{1}{R} \cdot \min_{k
\in [0,\frac{D}{5}]} \left(1-2s(k)+\frac{3^{k}}{R^{1/2}} \right) \mper$$
Substituting the previous equation in to
\eqref{eq:fraction-of-constraints}, we get that the fraction of
constraints satisfied by $\ell$ is at most
$$ \min_{k
\in [0,\frac{D}{5}]} \left(1-2s(k)+\frac{3^{k}}{R^{1/2}} \right) $$

\section{Hierarchy integrality  gaps for Unique Games and Related Problems} \label{sec:hierarchy}

This section is devoted to the construction of a integrality gap
instance for a hierarchy of SDP relaxations to Unique Games.  More
specifically, we consider the $\mathsf{LH}_{r}$ and $\mathsf{SA}_{r}$ SDP hierarchies
described in \cite{RaghavendraS09c}.  For these SDP hierarchies, we
will demonstrate the following integrality gap constructions.

\begin{theorem} \label{thm:sdp-hierarchy-gap-ug}
	For every $\epsilon,\delta > 0$, there exists an $\GF2^t\dashmaxtwolin$ instance $I$ for some positive
	integer $t$, such that no labelling
	satisfies more than $\delta$ fraction of edges of $\Gamma$ while
	there exists an SDP solution such that,
	\begin{itemize}
	\item the SDP solution is feasible for $\LH_{R}$ with $R =
		\exp(\exp(\Omega(\log\log^{1/2} N)))$.
	\item the SDP solution is feasible for $\SA_{R}$ with $R =
		\exp(\Omega(\log\log^{1/2} N))$.
	\item the SDP solution has value $1-O(\epsilon)$.
	\end{itemize}
	where $N$ is the number of vertices in the instance $I$.
\end{theorem}

\begin{remark}
Composing the above SDP integrality gap with Unique games based
hardness reductions yields corresponding gap instances for several
classes of problems like constraint satisfaction problems (CSPs) and ordering
CSPs like maximum acyclic subgraph.  Specifically, up to
$\exp(\exp(\Omega(\log\log^{1/2} N)))$ rounds of $\LH$ hierarchy or
the $\exp(\Omega(\log\log^{1/2} N)$ rounds of the $\SA$ hierarchy can
be shown to have te same SDP integrality gap as the
simple SDP relaxation for every CSP.
For the sake of brevity, we omit a formal statement of this result
here.
\end{remark}

Towards showing \pref{thm:sdp-hierarchy-gap-ug}, we follow the
approach outlined in \cite{RaghavendraS09c}.  At a high-level, the idea is to start with
an integrality gap instance $\Gamma$ for a simple SDP relaxation for
unique games over a large alphabet.  The instance $\Gamma$ is reduced
to an instance $\Psi_{\eps,Q,d}(\Gamma)$ of unique games over a smaller alphabet using a reduction similar to
Khot \etal \cite{KhotKMO07}.  Moreover, the SDP solution to the simple
SDP relaxation of $\Gamma$ can be translated to a solution for
several rounds of SDP hierarchy for $\Psi_{\eps,Q,d}(\Gamma)$.

Let $\Gamma$ be an instance of $\GF2^n\dashmaxtwolin$ over a set of
vertices $V(\Gamma)$ and edges $E(\Gamma)$.  On every
edge $(u,v) \in E(\Gamma)$, there is a constraint of the form $u - v =
\alpha_{uv}$ for some $\alpha \in \GF2^n$.  We will reduce $\Gamma$
to an instance of $Q\dashmaxtwolin$ instance using the two
query test described in \pref{sec:twoquerytest}.

\paragraph{Translations} Notice that Reed-Muller codes are invariant under translation of its
coordinates.  Therefore, the code $\rmdt$ and the test distributions
$\testt$ are both invariant under translation.  Formally, for an $\alpha \in \GF2^n$, the translation
operator $T_{\alpha} \from  Q^{N} \to Q^N$ is defined by
$$ (T_{\alpha} \circ c)_{\beta} = c_{\beta + \alpha}  \qquad \forall c
\in Q^{N}, \beta \in \GF{2}^n \mper$$
Given a codeword $c \in \rmdt$, we have $T_{\alpha} \circ c \in \rmdt$.

We are now ready to describe the reduction from $\Gamma$ to an
instance of $\GF2^n\dashmaxtwolin$.

\begin{mybox}
The vertices of $\Psi_{\eps,Q,d}(\Gamma)$ are $V(\Gamma) \times \rmdt$.  Let $\ell : V(\Gamma)
\times \rmdt \to Q$ be a labelling of the instance
$\Psi_{\eps,Q,d}(\Gamma)$.

\paragraph{Folding}  The labelling $\ell$ is assumed
to satisfy $\ell(v, c+r) = \ell(v, c) + r$  for every vertex $v \in
V(\Gamma), c \in \rmdt$ and $r \in Q$.  This is enforced
by ``folding''.

The constraints of $\Psi_{\eps,Q,d}(\Gamma)$ are given by the queries of the following verifier.
\begin{itemize}
\item Sample a vertex $u \in V(\Gamma)$ uniformly at random. Sample two
	neighbours $v_1,v_2 \in N(u)$ of $u$ uniformly at random. Let
	the constraint on the edge $(u,v_i)$ be $v_i - u = \alpha_i$
	for $i \in \{1,2\}$.
\item Sample an element $c_1 \in \rmdt$ uniformly at random, and
	sample a neighbour $c_2 \in \rmdt$ of $c_1$ in the graph
	$\Cay(\rmdt, \testt)$.

\item Sample an element $r \in Q$ uniformly at random.

\item Test if $\ell(v_1, (T_{\alpha_1} \circ c_1) + r) - r =
	\ell(v_2, T_{\alpha_2} \circ c_2) $.

\end{itemize}
\end{mybox}

\paragraph{Soundness}
\begin{lemma} \label{lem:ugreduction-soundness}
	For all sufficiently small constants $\epsilon, \delta > 0$
and all choices of $Q = 2^t$,
	there exists $\gamma, d$ such that if no labelling of $\Gamma$ satisfies more than
	$\gamma$ fraction of
	edges, then every labelling of $\Psi_{\eps,Q,d}(\Gamma)$
	satisfies at most $Q\Gamma_{\rho}\left(\nfrac{1}{Q}\right) + \delta$ fraction
	of constraints, where $\rho = e^{-\e}$.
\end{lemma}
\begin{proof}
Let $\ell : V \times \rmdt \to Q$ be a labelling of the instance
$\Psi_{\eps,Q,d}(\Gamma)$.  For each vertex $v \in V(\Gamma)$, let $F^{v} : \rmdt
\to Q$ denote the labelling $\ell$ restricted to the vertex
$v$, i.e., $F^{v}(c) \defeq \ell(v,c)$.  For each vertex $v \in V(\Gamma)$
and $q \in Q$ define $f^{v}_{q} : \rmdt \to [0,1]$ as
$$ f^{v}_{q}(c)  \defeq \Ind[F^{v}(c) = q] \mper $$
Due to folding we have $f^{v}_{q}(c) = f^{v}_{q+r}(c+r)$ for all $r
\in Q$.  Moreover, this implies that $\E_{c \in \rmdt} f^{v}_q =
\frac{1}{Q}$.
Finally, for a vertex $u \in V(\Gamma)$ and $r \in Q$ define,
$$ h^{u}_{r}(p) \defeq \E_{v \in N(u)} f^{v}_{r}(T_{\alpha_{uv}} \circ
p) \mper $$
Clearly, for the functions $h^{u}_{r}$ also we have,
\begin{equation}\label{eq:expectation-h}
	\E_{p} h^{u}_{r} = \frac{1}{Q} \qquad \forall u \in
	V(\Gamma), r \in Q
\end{equation}

The probability of acceptance of the verifier can be arithmetized in
terms of the functions $h^{u}_{r}$.
\begin{align*}
	& \Pr[\text{verifier accepts}] \\
	& =  \E_{u \in V(\Gamma)}
	\E_{v_1,v_2 \in N(u)} \E_{c_1, c_2 \in \Cay(\rmdt,\testt)
		} \E_{r \in Q} \left[\sum_{q \in
		Q} f^{v_1}_{q+r}(T_{\alpha_1} \circ c_1 +
		r)f^{v_2}_{q}(T_{\alpha_2} \circ c_2)\right] \\
	&= \E_{u \in V(\Gamma)}
		\E_{v_1,v_2 \in N(u)} \E_{c_1, c_2 \in \Cay(\rmdt,\testt)
		}   \left[\sum_{q \in
		Q} f^{v_1}_{q}(T_{\alpha_1} \circ
		c_1)f^{v_2}_{q}(T_{\alpha_2} \circ c_2)\right]
		\qquad (\text{folding})\\
	&= \E_{u \in V(\Gamma)} \E_{c_1, c_2 \in \Cay(\rmdt,\testt)
		}   \left[\sum_{q \in
		Q} \E_{v_1 \in N(u)} f^{v_1}_{q}(T_{\alpha_1} \circ
		c_1) \cdot \E_{v_2 \in N(u)} f^{v_2}_{q}(T_{\alpha_2}
		\circ c_2) \right]\\
	&= \E_{u \in V(\Gamma)} \E_{c_1, c_2 \in \Cay(\rmdt,\testt)
		}   \left[\sum_{q \in
		Q} h^{u}_{q}(c_1) h^{u}_{q}(c_2) \right]\\
	&= \E_{u \in V(\Gamma)}  \left[\sum_{q \in
	Q} \iprod{ h^{u}_{q},H h^{u}_{q}} \right]  \qquad
	(\text{where } H = \Cay(\rmdt,\testt))
\end{align*}
Suppose the probability of acceptance of the verifier is at least
$Q \cdot \Gamma_\rho(\nfrac{1}{Q}) + \delta$.  By simple averaging, for at
least $\delta/2$ fraction of the
vertices $u \in V(\Gamma)$ we have,
$$ \sum_{q \in
Q} \iprod{ h^{u}_{q},H h^{u}_{q}}  \geq
Q\Gamma_{\rho}(\nfrac{1}{Q}) + \frac{\delta}{2} \mper $$
Let us refer to such a vertex $u$ as being {\it good}.

Fix the parameters $\tau, L,d$ to those obtained by applying
\pref{thm:rm-majority-stablest-Q} with parameters $\epsilon,\delta/2Q$.
%
Recall that by \eqref{eq:expectation-h}, we have $\E_{\rmdt}[h^{u}_{q}] =
\frac{1}{Q}$.  Applying \pref{thm:rm-majority-stablest-Q}, if for each
$q \in Q$, $\max_{\alpha \in \GF2^{n}} \Inf_{\alpha}^{\leq
l}(h^{u}_{q}) \leq \tau$ then,
$$ \sum_{q \in Q} \iprod{ h^{u}_{q},G h^{u}_{q}}  \leq
Q\Gamma_{\rho}(\nfrac{1}{Q}) + Q \cdot \frac{\delta}{2Q} \mcom $$
This implies that for each {\it good} vertex $u$ there exists $q, \alpha$ such that $\Inf_{\alpha}^{\leq L}(h^{u}_q) \geq
\tau$.  We will use these influential coordinates to decode a
labelling for the $\GF2^n\dashmaxtwolin$ instance $\Gamma$.

For each vertex $v \in V(\Gamma)$ define the set of influential
coordinates $S_{v}$ as,
\begin{equation}
	S_v = \{ \alpha \in \GF2^{n} | \Inf^{\leq L}_{\alpha}(h^{v}_{q})
	\geq \tau/2 \text{for some } q \in Q \} \cup \{ \alpha
	\in \GF2^{n} | \Inf^{\leq L}_{\alpha}(f^{v}_{q})
	\geq \tau/2 \text{for some } q \in Q \}
\end{equation}
Using \pref{lem:sum-of-influences-bound}, for each of the functions $h^{v}_q$ or $f^{v}_q$, there are at most
$2L/\tau$ coordinates with influence greater than $\tau/2$.
Therefore, for each vertex $v$ the set $S_v$ is of size at most
$2\cdot Q \cdot 2L/\tau = 4QL/\tau$.

Define an assignment of labels $A : V(\Gamma) \to \GF2^n$ as
follows.  For each vertex $v$, sample a random $\alpha \in S_v$ and
assign $A(v) = \alpha$.

Fix one good vertex $u$, and a corresponding $q,\alpha$ such that
$\Inf_{\alpha}^{\leq L}(h_{q}^{u}) \geq \tau$.  By definition of
$h_{q}^{u}$ this implies that
$$ \Inf_{\alpha}^{\leq L} \left( \E_{v \in N(u)} T_{\alpha_{uv}}
\circ f^{v}_{q} \right) \geq \tau \mcom$$
which by convexity of influences yields,
$$ \tau \leq \E_{v \in N(u)} [ \Inf^{\leq
L}_{\alpha}(T_{\alpha_{uv}} \circ f^{v}_{q}) ] = \E_{v \in N(u)} [ \Inf^{\leq
L}_{\alpha - \alpha_{uv}}(f^{v}_{q}) ] \mper$$
Hence, for at least a $\tau/2$ fraction of the neighbours $v \in N(u)$,
the coordinate $\alpha - \alpha_{uv}$ has influence at least $\tau/2$
on $f^{v}_q$.  Therefore, for every good vertex $u$, for at least $\tau/2$ fraction
of its neighbours $v \in N(u)$, the edge $(u,v)$ is satisfied by the
labelling $A$ with probability at least $\frac{1}{|S_{u}|}
\frac{1}{|S_{v}|} \geq \tau^2/16Q^2 L^2$.  Since there are at
least a $\delta/2$-fraction of good vertices $u$, the expected fraction
of edges satisfied by the labelling $A$ is at least $\frac{\delta}{2}
\cdot \frac{\tau}{2} \cdot \frac{\tau^2}{16Q^2 L^2} =
\frac{\delta\tau^3}{64 Q^2 L^2}$.

By choosing the soundness $\gamma$ of the outer unique game $\Gamma$ to
be lower than $\frac{\delta\tau^3}{64 Q^2 L^2}$ yields a
contradiction.
This shows that the value of any labelling $\ell$ to
$\Psi_{\epsilon,Q,d}(\Gamma)$ is less than $Q \Gamma_{\rho}(1/Q) + \delta$.

\end{proof}

\paragraph{SDP Solution}
	We will construct feasible solutions to certain strong SDP
	relaxations of $\Psi_{\eps,Q,d}(\Gamma)$ by appealing to the work of \cite{RaghavendraS09c}.
	The SDP hierarchies that we consider are referred to as the
	$\LH$ and $\SA$ hierarchies.  Informally, the $r^{th}$-level
	$\LH$ relaxation ($\LH_r$) consists of the simple SDP relaxation for unique
	games augmented by local distributions $\mu_S$ over integral
	assignments for every set $S$ of at most $r$ vertices.  The
	local distribution $\mu_S$ is required to be consistent with
	the inner products of the SDP vectors.  Alternately, this SDP
	hierarchy can be thought of as, the simple SDP relaxation
	augmented by every valid constraint on at most $r$ vertices.

	The $\SA$ hierarchy is a somewhat stronger hierarchy that
	requires the local distributions $\mu_S$ to be consistent
	with each other, namely, $\mu_S$ and $\mu_T$ agree on $S \cap
	T$.  Alternately, the $\SA$ hierarchy corresponds to the
	simple SDP relaxation augmented with $r$-rounds of
	Sherali-Adams LP variables.  We refer the reader to \cite{RaghavendraS09c} for formal
	definitions of the $\SA$ and $\LH$ hierarchies.

	\begin{lemma} \label{lem:ugreduction-sdp-solution}
	Suppose $\Gamma$ has an SDP solution that of value $1-\eta$,
	then there exists an SDP solution to the instance
	$\Psi_{\eps,Q,d}(\Gamma)$ such that,
	\begin{itemize}
	\item the SDP solution is feasible for $\LH_{R}$ with $R =
		2^{\Omega(\epsilon/\eta^{1/4})}$.
	\item the SDP solution is feasible for $\SA_{R}$ with $R =
		\Omega(\epsilon/\eta^{1/4})$.
	\item the SDP solution has value $1-O(\epsilon)-o_{\eta}(1)$ on
		$\Psi_{\eps,Q,d}(\Gamma)$.
	\end{itemize}
\end{lemma}
\begin{proof}
This lemma is a direct consequence of Theorem $9$ from
\cite{RaghavendraS09c}.

In \cite{RaghavendraS09c}, the authors start with an integrality gap
instance $\Gamma$ for the simple SDP for unique games, and then
perform a traditional long code based reduction to obtain an instance
$\Phi_{\epsilon,Q}(\Gamma)$.

The crucial observation is the following.
\begin{observation} \label{obs:subset-instance}
	The vertices of $\Psi_{\eps,Q,d}(\Gamma)$ are a
	subset of vertices of $\Phi_{\epsilon,Q}(\Gamma)$ -- the instance obtained by the traditional $Q$-ary long code reduction on
	$\Gamma$.
\end{observation}
\begin{proof}
The vertices of $\Psi_{\eps,Q,d}(\Gamma)$ are pairs of
the form $(v,c)$ where $v \in V(\Gamma)$ and $c \in \rmdt$.  The
codeword $c \in \rmdt$ can be thought of as a string of length $N =
2^n$ over the alphabet $Q = \GF2^t$, namely, $c \in Q^{2^n}$.
The vertices of the instance $\Phi_{\epsilon,Q}(\Gamma)$ obtained via a traditional $Q$-ary long
code reduction is $V(\Gamma) \times Q^{2^n}$.  Hence the observation follows.
\end{proof}
In \cite{RaghavendraS09c}, the authors construct an SDP solution for
the instance $\Phi_{\epsilon,Q}(\Gamma)$ that is feasible for
$\LH_{R}$ relaxation with $R = 2^{\Omega(\epsilon/\eta^{1/4})}$ and
for $\SA_{R}$ relaxation with $R = \Omega(\epsilon/\eta^{1/4})$.
As noted in \pref{obs:subset-instance}, the vertices of $\Psi_{\epsilon,
Q}(\Gamma)$ are a subset of the vertices of
$\Phi_{\epsilon,Q}(\Gamma)$.  Therefore, the same SDP solution
constructed in \cite{RaghavendraS09c} when restricted to the instance
$\Psi_{\eps,Q,d}(\Gamma)$ yields a feasible solution for the
corresponding $\LH_R$ and $\SA_R$ relaxations.

To finish the proof, we need to show that the value of the SDP
solution from \cite{RaghavendraS09c} is $1-2\epsilon-o_{\eta}(1)$.

The traditional long code based reduction to get
$\Phi_{\epsilon,Q}(\Gamma)$ uses the noise stability test as the inner
gadget.  Namely, to test if a function $f: \GF{Q}^{2^n} \to \GF{Q}$ is
a dictator function, the verifier picks $x \in \GF{Q}^{2^n}$ uniformly
at random, and rerandomizes each coordinate of $x$ independently with
probability $\epsilon$, and then tests if $f(x) = f(y)$.  Composing
this noise stability test with the outer unique game $\Gamma$ yields
the instance $\Phi_{\epsilon,Q}(\Gamma)$.  The value of the SDP
solution constructed for $\Phi_{\epsilon,Q}(\Gamma)$ in
\cite{RaghavendraS09c} depends only on the expected hamming distance
between the queries $x,y$.   More precisely, in Claim $2$ of
\cite{RaghavendraS09c}, the authors show that if the distribution on the
queries $(x,y) \in \GF{Q}^{2^n}\times \GF{Q}^{2^n}$ is chosen to be an arbitrary distribution
$\NS$ over $\GF{Q}^{2^n}\times \GF{Q}^{2^n}$, the SDP
objective value of the solution is given by
$$ \Pr_{\set{x,y} \sim \NS,\ell \in [2^n]}[x_\ell = y_\ell]
-\epsilon \mper$$

The instance $\Psi_{\epsilon,Q,d}$ is obtained by using the
following distribution of $x,y$ over $\GF{Q}^{2^n} \times
\GF{Q}^{2^n}$ -- sample $(c_1,c_2)$ an edge in $\Cay(\rmdt,\testt)$.

By construction, for any coordinate $\ell \in [2^n]$, $\Pr[x_{\ell} =
y_{\ell} = 1-O(\epsilon)$.  Therefore,using Claim 2 of \cite{RaghavendraS09c}, the SDP objective value on the
instance $\Psi_{\epsilon,Q,d}(\Gamma)$ is at least $1-O(2\epsilon) - o_{\eta}(1)$.

\end{proof}

\begin{proof}[Proof of \pref{thm:sdp-hierarchy-gap-ug}]
Fix $t = \lceil 10/\epsilon \log (1/\delta) \rceil$ and $Q = 2^t$.  By
our choice of $Q$, we have $Q\Gamma_{e^{-\epsilon}}(1/Q) \leq \delta$ (see Appendix B in \cite{MosselOO05} for such asymptotic bounds on $\Gamma$).

Fix $\gamma,d$ depending on $\epsilon,\delta$ and $Q$ as dictated by
\pref{lem:ugreduction-soundness}.
Let $\Gamma$ be the Unique games instance obtained by
\pref{cor:code-to-uggap} with the optimal integral value set to
$\gamma$.
In particular, $\Gamma$ is a $\GF2^n \dashmaxtwolin$ instance that has $M
= 2^{2^{\log^2 n}}$ vertices. Its SDP optimum for the simple Unique games
SDP relaxation is at least $1- O(C(\epsilon,\delta)/n)$ ($\eta =
O(C(\epsilon,\delta)/n)$) for some constant $C(\epsilon,\delta)$
depending on $\epsilon,\delta$.

Now we apply the reduction to $\GF2^t \dashmaxtwolin$ outlined below
to obtain an instance $\Psi_{\eps,Q,d}(\Gamma)$.
The number of vertices of the instance $\Psi_{\eps,Q,d}(\Gamma)$ is
$|V(\Gamma)| \times |\rmdt|$.  Note that the choice of the degree
$d$ is a constant (say $d(\epsilon,\delta)$) depending on $\epsilon,\delta$.
Hence, the number of points in $\rmdt$ is given by $|\rmdt| =
2^{O(n^{d(\epsilon,\delta)})}$.  Therefore, the number of vertices of
$\Psi_{\epsilon,Q,d}(\Gamma)$ is $N = 2^{2^{\log^2 n}} \cdot
2^{O(n^{d(\epsilon,\delta)})} = 2^{2^{O(\log^2 n)}}$.
Equivalently, we have $n = 2^{\Omega(\log\log^{1/2} N)}$.

\begin{itemize}
	\item By \pref{lem:ugreduction-soundness}, the optimal labelling to
$\Psi_{\eps,Q,d}(\Gamma)$ satisfies at most
$Q\Gamma_{e^{-\epsilon}}(1/Q) + \delta = O(\delta)$ fraction of
constraints.

\item By \pref{lem:ugreduction-sdp-solution}, there exists an SDP solution
to the instance $\Psi_{\eps,Q,d}(\Gamma)$ with value $1- O(\epsilon)
- o_{\eta}(1)$.  Since $\eta = O(C(\epsilon,\delta)/n)$, for large
enough choice of $n$, the SDP value is at least $1-O(\epsilon)$.

\item The SDP solution is feasible for $\LH_R$ for $R =
2^{\Omega(\epsilon/\eta^{1/4})} =
2^{c(\epsilon,\delta)n^{1/4}} = \exp(\exp(\Omega(\log\log^{1/2} N)))$ rounds, where
$c(\epsilon,\delta)$ is a constant depending on $\epsilon$ and $\delta$.
Furthermore, the SDP solution is also feasible for $\SA_R$ for $R =
\Omega(\epsilon/\eta^{1/4}) = c(\epsilon,\delta)n^{1/4} =
\exp(\Omega(\log\log^{1/2} N)) $.
\end{itemize}

\end{proof}

\section*{Acknowledgements.}
We have had many beneficial conversations on this question with a number
of researchers including Venkatesan Guruswami, Ran Raz, Alex
Samorodnitsky, Amir Shpilka, Daniel Spielman, Ryan O'Donnell, Subhash Khot and Madhu
Sudan. We thank Avi Wigderson  for suggesting the name ``short
code''.

\addreferencesection
\bibliographystyle{amsalpha}
\bibliography{cayleysse}

\newpage

\appendix

\section{Missing Proofs}\label{sec:missingproofs}

\begin{proof}[Proof of \pref{thm:rm-majority-stablest}]
Suppose $d \geq C \log(1/\tau)$ for $C$ to be chosen later and fix $\gamma < 1/8$ for $\gamma$ to be chosen later. Let $\ell = \log(1/\tau)/4c_1< \tau^2 2^{d+1}$, where $c_1$ is the constant from \pref{thm:iprmbounded}. For $\alpha \in \GF2^N/ \cC$, let $\lambda_\alpha$ be the eigenvalues of $G$. Then, by \pref{lem:rm-eigenvalues},
\begin{equation}\label{eq:rmst}
  |\lambda_\alpha - \rho^k| < \tau, \text{ for }k \leq \ell,\;\;\;\;\; |\lambda_\alpha| < \rho^{\ell/2},\text{ for } k > \ell.
\end{equation}
Let $g=G^{\gamma} f$ and $G'=G^{1-2\gamma}$. Then, the graph $G'$ has the same eigenfunctions as $G$ - $\chi_\alpha$ for $\alpha \in \GF2^N/\cC$ with eigenvalues $\lambda_\alpha' = \lambda_\alpha^{1-2\gamma}$. From the above equation, it is easy to check that, for $\rho' = \rho^{1-2\gamma}$,
\begin{equation}\label{eq:rmst0}
  |\lambda_\alpha' - (\rho')^k| < \sqrt \tau, \text{ for }k \leq \ell,\;\;\;\;\; |\lambda_\alpha'| < (\rho')^{\ell/2},\text{ for } k > \ell.
\end{equation}
Further, as the eigenvalues of $G$ are each at most $1$, the coordinate influences of $g$ are no larger than those of $f$.

Now, decompose $g=g^{\le \ell} + g^{>\ell}$ into a low-degree part
  \begin{math}
    g^{\le \ell} %
    =\sum_{\alpha\in \GF2^n,~\wt(\alpha)\le \ell}\hat g(\alpha)\chi_\alpha
  \end{math}
  and a high-degree part
  \begin{math}
    g^{>\ell} %
    =\sum_{\alpha\in \GF2^n/\cC, ~\Delta(\alpha,\cC)>\ell} \hat
    g(\alpha)\chi_\alpha.
  \end{math}
  Then,%
  \begin{displaymath}
    \iprod{f,Gf} = \iprod{g,G'g}=\iprod{g^{\le \ell}, G' g^{\le \ell}}
    + \iprod{g^{>\ell},G' g^{>\ell}}
    \le \iprod{g^{\le \ell},G' g^{\le \ell}} +
    \mu\cdot \max_{\alpha\in\GF2^N/\cC,~\Delta(\alpha,\cC)>\ell} \lambda'_\alpha\mper
  \end{displaymath}
Hence, using Equation \eqref{eq:rmst0} (and the crude bound $\mu\le 1$),
  \begin{equation}
    \label{eq:rmst1}
    \iprod{f,G f} = \sum_{\sst{\alpha\in \GF2^N,\,\wt(\alpha)\le \ell}}
    (\rho')^{\wt(\alpha)} \hat g(\alpha)^2
    \quad + (\rho')^\ell + \sqrt\tau \mper
  \end{equation}
\eat{
  Let $P$ be the (formal) polynomial corresponding to $g^{\le \ell}$, that is,
  \begin{displaymath}
    P = \sum_{\sst{\alpha\in \GF2^N\\\wt(\alpha)\le \ell}} \hat g(\alpha)
    \prod_{\sst{i\in[N]\\\alpha_i=1}} u_i
    \mcom
  \end{displaymath}
  where $u_1,\ldots,u_N$ are formal variables. By construction, $P$ is a multilinear real polynomial with degree at most $\ell$.}

Observe that $g^{\le \ell}$ is a multilinear polynomial of degree at most $\ell$ and as the $\ell$-degree influences of $g$ are at most $\tau$, $g^{\le \ell}$ is $\tau$-regular.

Let $S\sse \sbits^N$ be the set of $\sbits$-vectors corresponding to the Reed--Muller code $\cC^\bot=\RM(n,d)$, that is, for every codeword $c\in \cC^\bot$, the set $S$ contains the vector $((-1)^{c_1},\ldots,(-1)^{c_N})$. Then, as $g$ is $[0,1]$-valued on $\rmd$ and $\zeta$ measures distance to bounded random variables, by Equation \eqref{eq:rmst},
\begin{multline*}\label{eq:rmst2}
\E_{z \sim S}[\zeta \circ g^{\le \ell}(z)] \leq \E_{z \sim S}[(g(z) - g^{\le \ell}(z))^2] = \E_{z \sim S}[(g^{> \ell}(z))^2] = \E_{z \sim S}[(G^\gamma f^{> \ell}(z))^2] \leq \\\max_{\alpha: |\alpha| > \ell} (\lambda_\alpha^{\gamma})^2 \leq \rho^{\gamma\ell}.
\end{multline*}

Hence, by \pref{thm:iprmbounded} (recall that $\ell = \log(1/\tau)/4c_1$),
  \begin{displaymath}
    \E_{x\sim \sbits^N} [\zeta \circ g^{\le \ell}(x)] \le \E_{z \sim S}[\zeta \circ g^{\le \ell}(z)] + 2^{O(\ell)} \sqrt \tau \leq \underbrace{ \rho^{\gamma \ell} + \tau^{1/4}}_{\eta \seteq} \mper
  \end{displaymath}
Now, as $\rmd$ is $\ell$-wise independent ($\ell < 2^{d+1}$),
\[ \E_{x \sim \sbits^N}[g^{\le \ell}(x)] = \E_{z \sim S}[g^{\le \ell}(z)] = \E_{z \sim S}[g(z)] \pm \E_{z \sim S}[(g^{> \ell}(z))^2]^{1/2} \leq \mu + \sqrt{\eta}.\]
Therefore, by \pref{cor:approx-majority-stablest}, 
  \begin{equation}\label{eq:rmst3}
    \iprod{g^{\le \ell},T_{\rho'} g^{\le \ell}} = \sum_{\alpha: \wt(\alpha) \le \ell} (\rho')^{\wt(\alpha)} \hat{g}(\alpha)^2 \le \Gamma_{\rho'}(\mu + \sqrt\eta) + \frac{O(\log\log(1/\eta))}{(1-\rho')\log(1/\eta)}.
  \end{equation}
Since $\Gamma_{\rho'}(\mu+ \sqrt\eta) \le \Gamma_{\rho'}(\mu) + 2\sqrt\eta$ and $\Gamma_\rho(\mu) \le \Gamma_{\rho'}(\mu) + |\rho - \rho'|/(1-\rho)$ (cf.~Lemma B.3, Corollary B.5 in \cite{MosselOO05}), it follows from Equations \eqref{eq:rmst1} and \eqref{eq:rmst3} that
  \begin{align*}
    \iprod{f,Gf} &= \iprod{g,G'g} \le \Gamma_\rho(\mu) + O\left(\frac{|\rho - \rho'|}{1-\rho}\right) + O(\sqrt\eta) + \frac{O(\log\log(1/\eta))}{(1-\rho)\log(1/\eta)} + \rho^{(1-2\gamma)\ell} + \sqrt\tau\\
&= \Gamma_\rho(\mu) + O\left( \frac{\gamma\log(1/\rho)}{1-\rho} + \rho^{\gamma \ell/2}+ \tau^{1/8}  + \frac{\log\log(1/\eta)}{(1-\rho)\log(1/\eta)}\right)\mper
  \end{align*}
(Here we used the estimate $|\rho - \rho'| = |\rho - \rho^{1-2\gamma}| = O(\gamma \log(1/\rho))$.) By choosing $d \geq C \log(1/\tau)$ and $\gamma = C K \log \log(1/\tau)/(\log(1/\tau)\log(1/\rho))$ for an appropriately large constant $C$, the above expression simplifies to
\[ \iprod{f,Gf} \leq \Gamma_\rho(\mu) + \frac{O(\log\log(1/\tau))}{(1-\rho)\log(1/\tau)}.\]
\end{proof}

\begin{proof}[Proof of \pref{lem:fooltobounded}]
  Since $X$ fools $\tau$-regular degree-$\ell$ PTFS, we have for all $u\ge 0$,
  \begin{displaymath}
    \Abs{\Prob{\zeta\circ Q(X)>u}-\Prob{\zeta\circ Q(Y)>u}} \le O(\e)\mper
  \end{displaymath}
  By hypercontractivity and $20\ell$-wise independence of $X$,
  \begin{displaymath}
    \Prob{\zeta\circ Q(X)>u}
    \le \Prob{\abs{Q(X)}>\sqrt u}
    \le u^{-10} \E Q(X)^{20}
    \le u^{-10} 2^{O(\ell)}\mper
  \end{displaymath}
  Since $\zeta\circ Q(X)$ is a non-negative random variable,
  \begin{displaymath}
    \E \zeta \circ Q(X) = \int \Prob{\zeta\circ Q(X) > u} \du
  \end{displaymath}
  Hence, we can bound its expectation
  \begin{align*}
    \E \zeta \circ Q(X) %
    &= \int_{u\ge 0} \Prob{\zeta\circ Q(X)> u} \du %
    \\
    & = \int_{0\le u\le M} \Prob{\zeta\circ Q(X)> u} \du %
    \quad \pm 2^{O(\ell)} \int_{u\ge M} u^{-10}\du \\
    & = \int_{0\le u\le M} \Prob{\zeta\circ Q(Y)> u} \du %
    \quad \pm O\Paren{ \e M + 2^{O(\ell )}/M^{9}} \\
    & = \E \zeta \circ Q(Y)%
    \quad \pm O\Paren{ \e M + \ell^{O(\ell d)}/M^{9}} \mper
  \end{align*}
  (In the last step, we used that $\Prob{ \zeta\circ Q(Y) > u} \le
  u^{-10} 2^{O(\ell)}$, a consequence of hypercontractivity.)
  Choosing $M=2^{O(\ell)}/\e^{0.1}$ (so that $\e M = 2^{O(\ell)}/
  M^{9}$), we conclude that $\E \zeta\circ Q = \E \zeta\circ Q \pm
  \e^{0.9} 2^{O(\ell)}$.
\end{proof}

\subsection{Proofs from \pref{sec:mosq}}\label{sec:mosqapp}
The following lemma shows a bound on the sum of influences.
\begin{lemma} \label{lem:sum-of-influences-bound}
  For a function $f\from\rmdt \to \R$ and $\ell<\dist(\cC^t)/2$, the sum of
  $\ell$-degree influences of $f$ is at most $\sum_{i\in[N]} \Inf^{\leq
    \ell}_{i}(f) \leq \ell\Var[f] $.
\end{lemma}

\begin{proof}
  The usual identity for the total (low-degree) influence holds,
  \begin{displaymath}
    \sum_{i\in[N]}\Inf^{\leq \ell}_i(f)
    =\sum_{\beta\in Q^N/\cC^t,~\wt(\beta)\le \ell}
    \wt(\beta)\hat f(\beta)^2\le \ell \Var f
    \mper\qedhere
  \end{displaymath}
\end{proof}

Analogous to \pref{thm:iprmbounded}, the following invariance principle can be
shown for regular multilinear polynomials.
\begin{theorem}\label{thm:iprmbounded-Q}
Let $N = 2^n$ and $t$ be an integer.  For every $\tau,\ell > 0$, there
exists $C$ such that for $d > C \log(1/\tau)$ the following holds:  if $P:\R^{Nt} \rightarrow
\R$ be a $\tau$-regular polynomial of degree at most $\ell$ then, for
$x \in_u \sbits^{Nt}$, $z \in_u \RM(n,d)^t$,
\[\left|\E[\zeta \circ P(x))] - \E[\zeta \circ P(z))] \right| \leq 2^{c_1 \ell} \sqrt \tau, \]
for a universal constant $c_1 > 0$.
\end{theorem}
The proof follows easily from the proof of Theorems \ref{thm:iprm} and \ref{thm:iprmbounded} and the fact that if $\RM(n,d)$ satisfies the properties of the PRG in \cite{MekaZ10}, then so does $\RM(n,d)^t$. We omit the proof.

The work of Mossel \etal \cite{MosselOO05} also obtains bounds on
noise stability of functions over product spaces of large alphabets
namely $Q^N$.  The following corollary is a consequence of Theorem $4.4$ in
\cite{MosselOO05}.  The proof is analgous to that of
\pref{cor:approx-majority-stablest} from \pref{thm:majority-stablest}.
\begin{corollary}
  \label{cor:approx-majority-stablest-Q}
  Let $f\from Q^N\to \R$ be a function with $\E f = \mu$ and $\E \zeta
  \circ f \le \tau$.
  Suppose $\Inf_i f^{\le 30\log(1/\tau)/\log Q}\le \tau$ for all $i\in[N]$.
  Then,
  \begin{displaymath}
    \iprod{f,T_\rho f}
    \le \Gamma_\rho(\mu)
    + O\left(\tfrac{\log Q\log\log(1/\tau)}{(1-\rho) \log
(1/\tau)}\right)
    \mcom
  \end{displaymath}
  where $T_\rho$ is the noise graph on $Q^N$ with second largest eigenvalue
  $\rho$ and $\Gamma_\rho$ is the Gaussian noise stability curve.
  (Here, we assume that $\tau$ is small enough.)
\end{corollary}

Now we are ready to present the proof of the majority is stablest
theorem over $\rmdt$ (\pref{thm:rm-majority-stablest-Q}) using
\pref{thm:iprmbounded-Q} and \pref{cor:approx-majority-stablest-Q}.

\begin{proof}[Proof of \pref{thm:rm-majority-stablest-Q}]
Let $Q = 2^t$.  Fix $d \geq C \log(1/\tau)$ for a sufficiently large
constant $C$ to be chosen later.
Let $\gamma < 1/8$ be a constant depending on $\epsilon, \delta$ whose value will be chosen later. Let $\ell =
\log(1/\tau)/4c_1< \tau^2 2^{d+1}$, where $c_1$ is the constant from
\pref{thm:iprmbounded-Q}. For $\alpha \in Q^N/ \cC$, let $\lambda_\alpha$ be the eigenvalues of $G$. Then, by \pref{lem:eigenvalues-larger-alphabet},
\begin{equation}\label{eq:rmst-Q}
	|\lambda_\alpha - \rho^{\wt(\alpha)}| < \tau, \text{ for }\wt(\alpha) \leq
  \ell,\;\;\;\;\; |\lambda_\alpha| < \rho^{\Omega(\ell/t)},\text{ for }
  \wt(\alpha) > \ell.
\end{equation}
Let $g=G^{\gamma} f$ and $G'=G^{1-2\gamma}$. Then, the graph $G'$ has
the same eigenfunctions as $G$ - $\chi_\alpha$ for $\alpha \in Q^N/\cC$ with eigenvalues $\lambda_\alpha' = \lambda_\alpha^{1-2\gamma}$. From the above equation, it is easy to check that, for $\rho' = \rho^{1-2\gamma}$,
\begin{equation}\label{eq:rmst0-Q}
	|\lambda_\alpha' - (\rho')^{\wt(\alpha)}| < \sqrt \tau, \text{
	for } \wt(\alpha) \leq \ell, \;\;\;\;\; |\lambda_\alpha'| <
	(\rho')^{\Omega(\ell/t)},\text{ for } \wt(\alpha) > \ell.
\end{equation}
Further, as the eigenvalues of $G$ are each at most $1$, the coordinate influences of $g$ are no larger than those of $f$.
Now, decompose $g=g^{\le \ell} + g^{>\ell}$ into a low-degree part
  \begin{math}
    g^{\le \ell} %
    =\sum_{\alpha\in Q^N,~\wt(\alpha)\le \ell}\hat g(\alpha)\chi_\alpha
  \end{math}
  and a high-degree part
  \begin{math}
    g^{>\ell} %
    =\sum_{\alpha\in Q^N/\cC, ~\wt(\alpha)>\ell} \hat
    g(\alpha)\chi_\alpha.
  \end{math}
  Then,%
  \begin{displaymath}
    \iprod{f,Gf} = \iprod{g,G'g}=\iprod{g^{\le \ell}, G' g^{\le \ell}}
    + \iprod{g^{>\ell},G' g^{>\ell}}
    \le \iprod{g^{\le \ell},G' g^{\le \ell}} +
    \mu\cdot \max_{\alpha\in Q^N/\cC,~\deg(\alpha)>\ell} \lambda'_\alpha\mper
  \end{displaymath}
Hence, using Equation \eqref{eq:rmst0-Q} (and the crude bound $\mu\le 1$),
  \begin{equation}
    \label{eq:rmst1-Q}
    \iprod{f,G f} = \sum_{\sst{\alpha\in Q^N/\cC,\,\wt(\alpha)\le \ell}}
    (\rho')^{\wt(\alpha)} \hat g(\alpha)^2
    \quad + (\rho')^{\Omega(\ell/t)} + \sqrt\tau \mper
  \end{equation}
\eat{
  Let $P$ be the (formal) polynomial corresponding to $g^{\le \ell}$, that is,
  \begin{displaymath}
    P = \sum_{\sst{\alpha\in \GF2^N\\\wt(\alpha)\le \ell}} \hat g(\alpha)
    \prod_{\sst{i\in[N]\\\alpha_i=1}} u_i
    \mcom
  \end{displaymath}
  where $u_1,\ldots,u_N$ are formal variables. By construction, $P$ is a multilinear real polynomial with degree at most $\ell$.}

Observe that $g^{\le \ell}$ is a multilinear polynomial of degree at
most $\ell \cdot t$.  Since the $\ell$-degree influences of $g$ are at
most $\tau$, it implies that the multilinear polynomial $g^{\le \ell}$ is $\tau$-regular.

Let $S\sse \sbits^{Nt}$ be the set of $\sbits$-vectors corresponding to
the Reed--Muller code $\rmdt$, that is, for every codeword
$c = (c^{(1)},c^{(2)},\ldots,c^{(t)}) \in \rmdt$, the set $S$ contains
the vector $((-1)^{c^{(1)}_1},\ldots
(-1)^{c^{(i)}_j},(-1)^{c^{(t)}_N})$. Then, as $g$ is $[0,1]$-valued on
$\rmdt$ and $\zeta$ measures distance to bounded random variables, by
Equation \eqref{eq:rmst-Q},
\begin{multline*}\label{eq:rmst2-Q}
\E_{z \sim S}[\zeta \circ g^{\le \ell}(z)] \leq \E_{z \sim S}[(g(z) -
g^{\le \ell}(z))^2] = \E_{z \sim S}[(g^{> \ell}(z))^2] = \E_{z \sim
S}[(G^\gamma f^{> \ell}(z))^2] \leq \\\max_{\alpha: \wt(\alpha) >
\ell} (\lambda_\alpha^{\gamma})^2 \leq \rho^{\Omega(\gamma\ell/t)}.
\end{multline*}

Hence, by \pref{thm:iprmbounded-Q} (recall that $\ell = \log(1/\tau)/4c_1$),
  \begin{displaymath}
    \E_{x\sim \sbits^N} [\zeta \circ g^{\le \ell}(x)] \le \E_{z \sim
    S}[\zeta \circ g^{\le \ell}(z)] + 2^{O(\ell)} \sqrt \tau \leq
    \underbrace{ \rho^{\Omega(\gamma \ell)} + \tau^{1/4}}_{\eta \seteq} \mper
  \end{displaymath}
Now, as $\rmdt$ is $\ell$-wise independent ($\ell < 2^{d+1}$),
\[ \E_{x \sim \sbits^N}[g^{\le \ell}(x)] = \E_{z \sim S}[g^{\le \ell}(z)] = \E_{z \sim S}[g(z)] \pm \E_{z \sim S}[(g^{> \ell}(z))^2]^{1/2} \leq \mu + \sqrt{\eta}.\]
Therefore, by \pref{cor:approx-majority-stablest-Q}, 
  \begin{equation}\label{eq:rmst3-Q}
    \iprod{g^{\le \ell},T_{\rho'} g^{\le \ell}} = \sum_{\alpha:
    \wt(\alpha) \le \ell} (\rho')^{\wt(\alpha)} \hat{g}(\alpha)^2 \le
    \Gamma_{\rho'}(\mu + \sqrt\eta) +
    \frac{O(t\log\log(1/\tau))}{(1-\rho')\log(1/\tau)}.
  \end{equation}
Since $\Gamma_{\rho'}(\mu+ \sqrt\eta) \le \Gamma_{\rho'}(\mu) +
2\sqrt\eta$ and $\Gamma_\rho(\mu) \le \Gamma_{\rho'}(\mu) + |\rho -
\rho'|/(1-\rho)$ (cf.~Lemma B.3, Corollary B.5 in \cite{MosselOO05}),
it follows from Equations \eqref{eq:rmst1-Q} and \eqref{eq:rmst3-Q} that
  \begin{align*}
    \iprod{f,Gf} &= \iprod{g,G'g} \le \Gamma_\rho(\mu) +
    O\left(\frac{|\rho - \rho'|}{1-\rho}\right) + O(\sqrt\eta) +
    \frac{O(t\log\log(1/\tau))}{(1-\rho)\log(1/\tau)} +
    \rho^{\Omega((1-2\gamma)\ell/t)} + \sqrt\tau
 \end{align*}
By a sufficiently small choice of  $\tau$, and fixing $\ell =
\log(1/\tau)/4c_1$ and $\gamma = 100 t c_1 \log
\log(1/\tau)/(\log(1/\tau)\log(1/\rho))$ (so that $\rho^{\Omega(\gamma
\ell/t)} <
1/\log (1/\tau)$ and $|\rho-\rho'| = O(\frac{t}{\log 1/\rho \log
1/\tau}))$,
the error term in the above expression can be made
smaller than $\delta$.


\end{proof}

\end{document}